\documentclass[10pt,twocolumn,twoside]{IEEEtran}

\usepackage{amsmath,amssymb,amsthm}
\usepackage{algorithmic}
\usepackage{amsmath,graphicx}
\usepackage{pgf,tikz,pgfplots}
\usepackage{bm}
\usepackage{color}
\usepackage{float}

\usepackage{subfigure}

\usepackage{tabularx}

\usepackage{hyperref}
\hypersetup{
	colorlinks=true,
	linkcolor=blue,
	filecolor=magenta,
	urlcolor=cyan,
	citecolor=blue
}

\newtheorem{theorem}{\bf Theorem}

\newtheorem{proposition}{\bf Proposition}

\usepackage{cases}

\newcommand{\textchange}[1]{{\color{black}{#1}}}

\renewcommand{\arraystretch}{1.1}
\usepackage{mathtools}
\usepackage{tcolorbox}
\usepackage{array}
\usepackage{epstopdf}
\newcolumntype{C}[1]{>{\centering\arraybackslash}m{#1}}

\newcommand{\figuresize}{2.2in}

\def\notvartualgraph{0}
\ifthenelse{\notvartualgraph=1}{
	\def\barr#1{{{#1}}}
	\def\xhatzero{x}
}
{
	\def\barr#1{{\bar{#1}}}
	\def\xhatzero{x^0}
}

\begin{document}

\title{Optimal Network Slicing for Service-Oriented Networks with Flexible Routing and Guaranteed E2E Latency\thanks{Part of this work \cite{Chen2020} has been presented at the 21st IEEE International Workshop on Signal Processing Advances in Wireless Communications (SPAWC), Atlanta, Georgia, USA, May 26--29, 2020.}}

\author{\IEEEauthorblockN{Wei-Kun Chen, Ya-Feng Liu, Antonio De Domenico, Zhi-Quan Luo, and Yu-Hong Dai}
	\thanks{The work of W.-K. Chen was supported in part by Beijing Institute of Technology Research Fund Program for Young Scholars.
		The works of Y.-F. Liu and Y.-H. Dai were supported in part by the National Natural Science Foundation of China (NSFC) under Grant 12022116, Grant 12021001, Grant 11688101, and Grant 11991021, Grant 11631013, and the Strategic Priority Research Program of Chinese Academy of Sciences, Grant XDA27000000.
		The work of Z.-Q. Luo was supported in part by the National Natural Science Foundation of China (No. 61731018) and the Shenzhen Fundamental Research Fund (No. KQTD201503311441545).
		(\emph{Corresponding author: Ya-Feng Liu.})}
	\thanks{W.-K. Chen is with the School of Mathematics and Statistics/Beijing Key Laboratory on MCAACI, Beijing Institute of Technology, Beijing 100081, China (e-mail: chenweikun@bit.edu.cn).
		Y.-F. Liu {and Y.-H. Dai} are with the State Key Laboratory of Scientific and Engineering Computing, Institute of Computational Mathematics and Scientific/Engineering Computing, Academy of Mathematics and Systems Science, Chinese Academy of Sciences, Beijing 100190, China (e-mail: \{yafliu, dyh\}@lsec.cc.ac.cn).
		A. De Domenico is with the Huawei Technologies Co. Ltd, France Research Center, 92100 Boulogne-Billancourt, France (e-mail: antonio.de.domenico@huawei.com).
		Z.-Q. Luo is with the Shenzhen Research Institute of Big Data and The Chinese University of Hong Kong, Shenzhen 518172, China (e-mail: luozq@cuhk.edu.cn)
	}
}


\maketitle

\begin{abstract}
	Network function virtualization is a promising technology to simultaneously support multiple services with diverse characteristics and requirements in the 5G and beyond networks.
	In particular, each service consists of a predetermined sequence of functions, called
	service function chain (SFC), running on a cloud environment.
	To make different service slices work properly in harmony, it is crucial to
	appropriately select the cloud nodes to deploy the functions in the SFC and flexibly
	route the flow of the services such that these functions are processed in the order
	defined in the corresponding SFC, the end-to-end (E2E) latency constraints of all
	services are guaranteed, and all cloud and communication resource budget constraints are
	respected.
	In this paper, we {first} propose a new mixed binary linear program (MBLP) formulation of the above network slicing problem that optimizes the system energy efficiency while jointly considers the E2E latency requirement, resource budget, flow routing, and functional instantiation.
	{Then, we develop another MBLP formulation and show that the two formulations are equivalent in the sense that they share the same optimal solution. However, since the numbers of variables and constraints in the second problem formulation are significantly smaller than those in the first one, solving the second problem formulation is more computationally efficient especially when the dimension of the corresponding network is large.}
	Numerical results demonstrate the advantage of the proposed formulations compared with the existing ones.
\end{abstract}

\IEEEpeerreviewmaketitle

\section{Introduction}\label{sec:introduction}
Network function virtualization (NFV) is considered as one of the key technologies for the fifth generation (5G) and beyond 5G (B5G) networks \cite{Mijumbi2016}.
In contrast to traditional networks where service functions are processed by dedicated {hardware} in fixed locations, NFV can efficiently take the advantage of cloud technologies to configure some specific nodes in the network to process network service functions on-demand, and then flexibly establish a customized virtual network for each service request.
In the NFV-enabled network, classical networking nodes are integrated with NFV-enabled nodes (i.e., cloud nodes) and each service consists of a predetermined sequence of virtual network functions (VNFs), called service function chain (SFC) \cite{Zhang2013,Halpern2015,Mirjalily2018}, which can only be processed by certain specific cloud nodes \cite{Zhang2017,Baumgartner2015,Oljira2017}.
In practice, each service flow has to pass all VNFs in its SFC in sequence and its end-to-end (E2E) latency requirement must be satisfied.
However, since all VNFs run over a shared common network infrastructure, it is crucial to allocate network (e.g., cloud and communication) resources to meet the diverse service requirements, subject to {the SFC constraints, the E2E latency constraints of all services}, and all cloud nodes' and links' capacity constraints.
{We call the above resource allocation problem {\emph{network slicing for service-oriented networks}} (\emph{network slicing}\footnote{{Notice that in the 5G network slicing, each slice can be a service function chain or a virtual network (a network that is not a chain). In this paper, we restrict to study the case where all slices are service function chains.}} for short in this paper).} 
\subsection{Related Works}

{Considerable works have been done on network slicing recently; see \cite{Zhang2017}-\cite{Carpio2017} and the references therein.}
More specifically, {references} \cite{Zhang2017} and \cite{Zhang2019} considered the VNF deployment problem with a limited network resource constraint.
Reference \cite{Liu2017} considered the joint problem of new service function chain deployment and in-service function chain readjustment.
{However, references \cite{Zhang2017,Zhang2019}, and \cite{Liu2017}} did not take the E2E latency constraint of each service into consideration, which is one of the key design considerations in the 5G network \cite{Agyapong2014}.
Reference \cite{Domenico2019} investigated a specific two-layer network which consists of a  {central} cloud node and several edge cloud nodes without considering the limited link (bandwidth) capacity constraint.
Reference \cite{Trivisonno2015} presented a formulation with the E2E latency requirement for the virtual network embedding problem in the {5G} systems, again without considering the limited node (computational) capacity constraint.
Obviously, the solution obtained in \cite{Domenico2019} and \cite{Trivisonno2015} without considering the limited link/node capacity constraints in the corresponding problem formulation may lead to violations of resource constraints. 
Reference \cite{Luizelli2015} considered the joint placement of VNFs and routing of traffic flows between the data centers that host the VNFs and proposed to minimize the number of deployed VNFs under latency constraints.
Reference \cite{Jiang2012} studied the data center traffic engineering problem and again emphasized the importance of the joint placement of virtual machines and routing of traffic flows between the data centers hosting the virtual machines.
Reference \cite{Guo2011} investigated the
virtual network embedding problem of shared backup network provision.
However, the above references \cite{Luizelli2015}, \cite{Jiang2012}, and \cite{Guo2011} simplified the routing strategy by selecting paths from a predetermined path set, which may possibly degrade the overall performance. 
Reference \cite{Narayana2013} limited the routing strategy to be one-hop routing.
Reference \cite{Zhang2015} considered a simplified setup where there is only a single function in each SFC.
{References \cite{Addis2015} and \cite{Basta2017} simplified the VNF placement decision-making by assuming that all VNFs in an SFC must be instantiated at the same cloud node.}
Reference \cite{Woldeyohannes2018} proposed a way of analyzing the dependencies between traffic routing and VNF placement in the NFV networks.
Reference \cite{Gouareb2018} studied the problem of placement of VNFs and routing of traffic flows to minimize the overall latency. 
A common assumption in \cite{Baumgartner2015,Oljira2017,Jiang2012,Woldeyohannes2018}, and \cite{Gouareb2018} is that only a single path was allowed to transmit the data flow of each service.
Apparently, formulations based on such assumptions do not fully exploit the flexibility of traffic routing {(i.e., allow the traffic flow to split into multiple paths)} and hence might affect the performance of the whole network.
{We remark that flexible traffic flow routing may cause out-of-order data rate delivery in the physical network. 
However, this can be easily resolved using hash-based splitting; see \cite{Yu2008} for more details. 
Indeed, flexible traffic flow routing is realizable in many modern software defined network (SDN) architectures \cite{Paschos2018}.}
References \cite{Narayana2013} and \cite{Xu2013}-\cite{Carpio2017} assumed that instantiation of a VNF can be split over multiple cloud nodes, which may result in high coordination overhead in practice.

In a short summary, the existing works on the network slicing problem {either do not} consider the E2E latency constraint of each service (e.g., \cite{Zhang2017,Zhang2019,Liu2017}), or do not consider the cloud and communication resource budget constraints (e.g., \cite{Domenico2019,Trivisonno2015}), or simplify the routing strategy by selecting paths from a predetermined path set (e.g.,  \cite{Luizelli2015}-\cite{Guo2011}), or enforce that each flow can only be transmitted via a single path (e.g., \cite{Baumgartner2015}, \cite{Oljira2017}, \cite{Jiang2012}, \cite{Woldeyohannes2018}, \cite{Gouareb2018}), or make impractical assumptions on function initialization (e.g., \cite{Narayana2013}, \cite{Xu2013}-\cite{Carpio2017}).
To the best of our knowledge, for the network slicing problem, none of the existing formulations/works simultaneously takes all of the above practical factors (e.g., E2E latency, resource budget, flexible routing, and coordination overhead) into consideration.
The goal of this work is to fill this research gap, i.e., provide mathematical formulations of the network slicing problem that {\emph{simultaneously}} allow the traffic flows to be flexibly transmitted on (possibly) multiple paths, satisfy the E2E latency requirements of all services and all cloud nodes' and links' capacity constraints, and require that each service function in an SFC is processed by exactly one cloud node.
\subsection{Our Contributions}
In this paper, we propose two new mathematical formulations of the network slicing problem which simultaneously take the E2E latency requirement, resource budget, flow routing, and functional instantiation into consideration.
The main contributions of this paper are summarized as follows.

\begin{itemize}
	\item 
	      By integrating the traffic routing flexibility into the formulation in \cite{Woldeyohannes2018}, we first propose a mixed binary \emph{linear} programming (MBLP) formulation (see problem \eqref{mip} further ahead), which is natural (in the terms of its design variables) and can be solved by standard solvers like Gurobi \cite{Gurobi}.
	      The formulated problem minimizes a weighted sum of the total power
	      consumption of the whole {cloud} network (equivalent to the total number of activated cloud nodes) and the total delay of all services subject to {the SFC constraints, the E2E latency constraints of all services, and all cloud nodes' and links' capacity constraints}.
	      {Notice that minimizing the total delay of all services is advantageous for some delay critical tasks as well as avoiding cycles in the traffic flow.}
	\item 
	      {Since the numbers of variables and constraints are huge in the above formulation, we then develop an equivalent MBLP reformulation (in the sense that the two formulations share the same optimal solutions) whose numbers of variables and constraints are significantly smaller (see problem \eqref{newmip} further ahead), which makes it more efficient to solve the network slicing problem especially when the dimension of the network is large.}
\end{itemize}
\noindent Simulation results demonstrate the efficiency and effectiveness of the proposed formulations. More specifically, our simulation results show: 1) the compact formulation significantly outperforms the natural formulation in terms of the solution efficiency and is able to solve problems in a network with realistic dimensions;
2) {our proposed formulations are more effective than the existing formulations in \cite{Zhang2017} and \cite{Woldeyohannes2018} in terms of the solution quality and are able to flexibly route the traffic flows and guarantee the E2E latency of all services.}

The paper is organized as follows. Section \ref{sec:modelformulation} first introduces the system model, followed by an illustrative example that motivates this work.
Section \ref{naturalformulation} presents a natural formulation for the network slicing problem and Section \ref{compactformulation} presents a more compact problem formulation.
Section \ref{subsec:experiments} reports the computational results.
Finally, Section \ref{conclusions} draws the conclusion.

\section{Problem Statement}
\label{sec:modelformulation}

\subsection{System Model}
\label{subsec-systemmodel}
Consider a directed network $\mathcal{G}=\{\mathcal{I},\mathcal{L}\}$, where $\mathcal{I}=\{i\}$ is the set of nodes and $\mathcal{L}=\{(i,j)\}$ is the set of links. 
We require that the total data rate on each link $(i,j)$ is upper bounded by the capacity $C_{i,j}$.
Due to this, the queuing delay on each link can be assumed to be negligible; see \cite{Liu2016}.
As a result, we can assume that the expected (communication) delay on each link is equal to the propagation delay, which is known as a constant $ d_{i,j} $ \cite{Luizelli2015,Woldeyohannes2018,Mohammadkhan2015}.
Let $ \mathcal{V} $ be a subset of $ \mathcal{I} $ denoting the set of the cloud nodes.
Each cloud node $ v $ has a computational capacity $ \mu_v $ and we assume as in \cite{Zhang2017} that processing one unit of data rate requires one unit of (normalized) computational capacity.
The network supports a set of flows $\mathcal{K}=\{k\}$.
Let $S(k)$ and $D(k)$ be the source and destination nodes of flow $k$, respectively, and suppose that $S(k),D(k)\notin \mathcal{V}$.
Each flow $ k $ relates to a distinct service, which is given by an SFC consisting of $ \ell_k $ service functions that have to be performed in sequence by the network:
\begin{equation}
	\label{sequence}
	f_{1}^k\rightarrow f_{2}^k\rightarrow \cdots \rightarrow f_{\ell_k}^k.
\end{equation}
As required in \cite{Zhang2017,Domenico2019}, and \cite{Woldeyohannes2018}, to minimize the coordination overhead, each function must be instantiated at exactly one cloud node.
If function $ f^k_s $, $ s \in \mathcal{F}(k) := \{1,\ldots, \ell_k\} $, is processed by cloud node $ v $ in $ \mathcal{V} $, \textchange{we assume that  the expected NFV delay depends on the data rate consumed by function $ f^k_s $ and the hardware of the cloud node (e.g., CPU, memory), and through the pre-computation it can be seen as a constant $	d_{v,s}(k) $; see \cite{Luizelli2015} and \cite{Woldeyohannes2018}.}
{The service function rate that no function of flow $k$ has been processed is denoted as $\lambda_0(k)$. Similarly, the service function rate that function $f_s^k$ has just been processed (and function $f_s^{k+1}$ has not been processed) is denoted as $\lambda_{s}(k)$.}
Each flow $ k $ is required to have an {E2E latency} guarantee, denoted as $ \Theta_k $.
\subsection{The Network Slicing Problem and an Illustrative Example}
As all service functions run over a shared common network infrastructure,
the network slicing problem aims to allocate cloud and communication resources
and to determine functional
instantiation of all flows and the routes and associated
data
rates of all flows on the corresponding routes to meet diverse service
requirements.
In order to obtain a satisfactory solution, it is crucial to establish a problem
formulation that jointly takes various practical factors, especially flexible
routing and E2E delay, into consideration.
Indeed, flexible routing, as used in \cite{Zhang2017,Zhang2019}, and \cite{Liu2017}, allows the
traffic flows to flexibly select their routes and associated data rates on the
corresponding routes according to the network infrastructure (e.g., links'
capacities), and thus can possibly improve the solution quality (as compared with the routing strategy of selecting paths from a predetermined path set or
enforcing each flow to transmit on only a single path).
In addition, delay is one of the key metrics in the 5G networks \cite{Agyapong2014} and
a virtualized communication system requires the E2E delays of all services to
be below given thresholds \cite{Domenico2019}.
Next, we shall use an illustrative example to show how flexible routing and E2E latency affect the solution of the network slicing problem.

\begin{figure}[h]
	\centering
	\begin{tikzpicture}[scale=\textwidth/20cm]
		\draw [->,line width=0.8pt] (1.6,5.4) -- (3.85,7.7);
		\draw [->,line width=0.8pt] (1.6,5.4) -- (3.95,3.05);
		\draw [->,line width=0.8pt] (3.93,7.77) -- (8.67,5.5);
		\draw [->,line width=0.8pt] (4,3) -- (8.65,5.38);
		\draw [->,line width=0.8pt] (4,3) -- (3.92,7.7);
		\draw [->,line width=0.8pt] (8.71,5.43) -- (6.05,5.43);
		\draw [->,line width=0.8pt] (5.97,5.35) -- (4,7.7);
		\begin{small}
			\draw [fill=black] (1.6,5.4) circle (2.5pt);
			\draw[color=black] (1.5,5.7) node {A};
			\draw [fill=black] (3.93,7.77) circle (2.5pt);
			\draw[color=black] (4.2,7.9) node {B};
			\draw [fill=black] (4,3) ++(-2.5pt,0 pt) -- ++(2.5pt,2.5pt)--++(2.5pt,-2.5pt)--++(-2.5pt,-2.5pt)--++(-2.5pt,2.5pt);
			\draw[color=black] (4.35,3.5) node {C(4)};
			\draw [fill=black] (5.97,5.43) circle (2.5pt);
			\draw[color=black] (6,5.7) node {D};
			\draw [fill=black] (8.71,5.43) ++(-2.5pt,0 pt) -- ++(2.5pt,2.5pt)--++(2.5pt,-2.5pt)--++(-2.5pt,-2.5pt)--++(-2.5pt,2.5pt);
			\draw[color=black] (8.8,5.8) node {E(8)};
			\draw[color=black] (2.5,6.8) node {(2,1)};
			\draw[color=black] (3.2,4.3) node {(2,1)};
			\draw[color=black] (6.6,6.8) node {(2,1)};
			\draw[color=black] (6.1,4.4) node {(2,1)};
			\draw[color=black] (4.4,5.5) node {(2,1)};
			\draw[color=black] (7.3,5.6) node {(4,1)};
			\draw[color=black] (5.4,6.6) node {(2,1)};
			\draw (6.95,3.26) node[anchor=north west] {{Cloud nodes}};
			\draw [fill=black] (6.8,3.0) ++(-2.5pt,0 pt) -- ++(2.5pt,2.5pt)--++(2.5pt,-2.5pt)--++(-2.5pt,-2.5pt)--++(-2.5pt,2.5pt);
			\draw (0,2.76) node[anchor=north west] {{(a,b) over each link: the link capacity is $a$ and the communication}};
			\draw (0,2.26) node[anchor=north west] {{~~~~~~~~~~~~~~~~~~~~~~~~~delay is $b$ over this link}};
			\draw (0,1.76) node[anchor=north west] {{(c) over each node: the node capacity is $c$}};
		\end{small}
	\end{tikzpicture}
	\caption{A network example.}
	\label{originalmap}
\end{figure}

Consider the network example in Fig. \ref{originalmap}. As shown in
Fig. \ref{originalmap}, the computational capacities on cloud nodes C and
E are $4$ and $8$, respectively, and all links' capacities are $2$ except link $(\text{E},\text{D})$ whose link
capacity is $4$. The communication delay on each link $ (i,j) $ is $ d_{i,j}=1 $.
There are two different functions, i.e., $ f^1 $ and $ f^2 $.
Cloud node C can only process function $ f^2 $, while cloud node E can {process} both functions $ f^1 $ and $ f^2 $.
The NFV delays of both functions at (possible) cloud node C and cloud node E are $1$.

\emph{Flexible routing allows the traffic flows to transmit over possible multiple paths and thus alleviates the effects of (low) link capacities on the network slicing problem.}
Suppose that there is only a single service from node A to node D with the E2E delay threshold being $ \Theta_{1}=5 $.
{The SFC considered is $f_1 \rightarrow f_2 $} and all
the service function rates $ \lambda_0(1) $, $ \lambda_1(1) $, and $\lambda_2(1)$ are $ 4 $.
If only a single path is allowed to transmit the traffic flow (as in \cite{Baumgartner2015,Oljira2017,Jiang2012,Woldeyohannes2018}, and \cite{Gouareb2018}), no solution exists for this example due to the \emph{limited} link capacity.
Indeed, either link $(\text{A},\text{B})$'s or link $(\text{A},\text{C})$'s capacity is 2, which is not enough to support a traffic flow with a data rate being $4$.
However, in sharp contrast, if the traffic flow can be flexibly transmitted on multiple
paths, a feasible solution is given as follows: first use paths
$ \text{A} \rightarrow \text{B} \rightarrow \text{E} $ and
$ \text{A} \rightarrow \text{C} \rightarrow \text{E} $ to simultaneously route
the flow from node A to node E where the data rates on both paths are
$ 2 $; after functions $ f^1 $ and $f^2$ being processed by node E, route the flow
to the destination node D using link $  (\text{E}, \text{D})$.
For this solution, the communication delays from node A to node E and node
E to node D are $\max\{ d_{\text{A},\text{B}} + d_{\text{B},\text{E}},d_{\text{A},\text{C}} + d_{\text{C},\text{E}}\} = 2$
and $d_{\text{E},\text{D}}=1$, respectively.
Thus, the total communication delay is $3$.
In addition, as functions $f^1$ and $f^2$ are hosted at node E, the total NFV delay is $2$.
Then, the E2E delay is equal to the sum of the total communication and NFV delays, which is $5$, implying that such a solution satisfies the E2E latency requirement of the service.
This clearly shows the benefit of flexible routing in the network slicing
problem, i.e., it alleviates the effects of (low) links' capacities to support the services.

\emph{The E2E latency constraints of all services need to be explicitly enforced in the problem formulation.}
Suppose there are in total two services where service I is from node A to node D with the E2E delay threshold being $\Theta_{1}=4$ and service II is from node A to node B with the E2E delay threshold being $\Theta_2=3$.
Functions $ f^1 $ and $ f^2 $ need to be processed for services I and II, respectively; for each service $ k $, the service function rates $ \lambda_0(k) $ and $ \lambda_1(k) $ are $1$.
Our objective is to minimize the number of activated cloud nodes, because this reflects the total energy consumption in the network (as shown in Section \ref{subsec:mipformulation}). 
Suppose that the {E2E} latency constraint of each service is not enforced (as in \cite{Zhang2017}, \cite{Zhang2019}, and \cite{Liu2017}).
Then, the optimal solution is that both functions are processed by cloud node E as follows:\vspace{0.1cm}\\
$\begin{array}{l}
		\text{Service I}: \text{A} \rightarrow \text{B} \rightarrow \text{E~{(providing function $ f^1 $)}} \rightarrow  \text{D},                                       \\
		\text{Service II}: \text{A} \rightarrow \text{C} \rightarrow \text{E~{(providing function $ f^2 $)}} \rightarrow  \text{D}  \rightarrow \text{B}. \vspace{0.1cm} \\
	\end{array}$\\
For service II, it traverses $ 4 $ links from node A to node B with a total communication delay being $ 4 $, 	which, pluses the NFV delay $1$, obviously violates its E2E latency constraint.
Therefore, to obtain a better solution, it is necessary to enforce the E2E latency constraints of the two services explicitly in the problem formulation.
Then, the solution of the problem with the E2E latency constraints is:\vspace{0.1cm}\\
$\begin{array}{l}
		\text{Service I}: \text{A} \rightarrow \text{B} \rightarrow \text{E~{(providing function $ f^1 $)}} \rightarrow  \text{D}, \\
		\text{Service II}: \text{A} \rightarrow  \text{C~{(providing function $ f^2 $)}} \rightarrow  \text{B}.
	\end{array}$\vspace{0.1cm}\\
In this solution, the E2E delays of service I and service II in the above solution are $ 4 $ and $ 3 $, respectively, which satisfy the E2E latency requirements of both services.

In summary, the example in Fig. \ref{originalmap} illustrates that, in order to
obtain a satisfactory solution to the network slicing problem, it is crucial to allow
the flexible routing and enforce the E2E latency constraints of all services
explicitly in the problem formulation.

\section{Problem Formulation}
\label{naturalformulation}
\subsection{Preview of the Problem Formulation}
The network slicing problem is to determine functional instantiation of all flows and the routes and associated data rates of all flows on the routes while satisfying the SFC requirements, the E2E delay requirements, and the capacity constraints on all cloud nodes and links.
In this section, we shall provide a new problem formulation of the network slicing problem which takes practical factors like flexible routing and E2E latency requirements into consideration; see problem \eqref{mip} further ahead.

Our proposed formulation builds upon those in two closely related works \cite{Woldeyohannes2018} and \cite{Zhang2017} but takes further steps.
More specifically, in sharp contrast to the formulation in \cite{Woldeyohannes2018} where only a \emph{single} path is allowed to route the traffic flow of each service (between two cloud nodes processing two adjacent functions of a service), our proposed formulation allows the traffic flow of each service to transmit on (possibly) multiple paths and hence fully exploits the flexibility of traffic routing; different from that in \cite{Zhang2017}, our formulation guarantees the E2E delay of all services, which consists of two types of delays: total communication delay on the links and total NFV delay on the cloud nodes. 

Next, we describe the constraints and objective function of our formulation in details.

\subsection{Various Constraints}
\label{subsec:equivalenttrans}
In this subsection, we shall present various constraints of the network slicing problem. Before doing it, we first present an equivalent virtual network that plays an important role in presenting the constraints.
{\bf\vspace{0.1cm}\\$\bullet$ An Equivalent Virtual Network\vspace{0.1cm}}
	
	References \cite{Zhang2017,Xu2017}, and \cite{Cheng2018} assume that each cloud 
	node can process at most one function of the same flow in the physical network.
	This assumption enforces that different functions of each flow must be hosted at
	different cloud nodes, which thus potentially increases the number of cloud nodes
	needed to be activated {(and therefore the power consumption in the cloud network)}
	and the total communication delay of the flow (as the flow
	needs to traverse more links).
	{Therefore, in this paper, we do not impose such an assumption, i.e., we allow that multiple functions of the same flow are processed by the same cloud node. To do this, we introduce an equivalent virtual network in the following.}
	In the next, we shall call the original network as the physical network to distinguish it from the constructed virtual network.
	
	We construct the virtual network as follows. Let $ \mathcal{\bar{G}}= (\mathcal{\bar{I}}, \mathcal{\bar{L}}) $ denote the virtual network and $ \mathcal{\bar{V}} $ denote the set of the cloud nodes in the virtual network.
	We first construct $\mathcal{{\bar{V}}}$ and $\mathcal{\bar{I}}$.
	Let $n_v$ be the number of functions that (physical) cloud node $v$ can process.
	Denote $\ell_{\max}$ be the maximum number of functions in an SFC among all flows, i.e., $\ell_{\max}=\max_{k\in \mathcal{K}}{\ell_k}$.
	{Then,} the maximum number of functions that can be possibly hosted at (physical) cloud node $v$ for each flow is $m_v = \min \{ n_v , \ell_{\max} \}$.
	For each cloud node $v \in \mathcal{V} $ in the physical network, we first set $ v $ as \textchange{an intermediate} node (i.e., a node that can route flows but cannot
	process any service function) and then introduce $ m_v $ virtual cloud nodes, namely, $ \mathcal{I}_{v}= \{v_1, \ldots, v_{m_v}\}  $.
	Then, the sets of cloud nodes and nodes in the virtual network are defined as $ \mathcal{\bar{V}}=   \mathcal{I}_1 \cup \cdots \cup \mathcal{I}_{|\mathcal{V}|} $ and $ \mathcal{\bar{I}} =\mathcal{I} \cup \mathcal{\bar{V}} $, respectively.
	Next, we construct $\bar{\mathcal{L}}$. First, $ \mathcal{\barr{L}} $ contains all links in $ \mathcal{L} $.
	In addition, for each $ v \in \mathcal{V} $ and $ 1 \leq t \leq m_v $, we construct the links $ (v,v_t) \in \mathcal{\bar{L}}  $ and $  (v_t, v) \in \mathcal{\bar{L}}  $.
	Therefore, each virtual cloud node is associated with exactly two links in the virtual network.
	We now specify the cloud nodes' and links' capacities and delays in the virtual network.
	Since each (virtual) cloud node $ v_t $, $ t\in \{1, \ldots, m_v \} $, is a copy of (physical) node $ v $, the NFV delay of function $ f_s^k $ on it is the same as that on (physical) node $ v $, i.e., $ d_{v,s}(k) $; the sum of the computational capacities over all (virtual) nodes $ v_1, \ldots, v_{m_v} $ is $ \mu_v $.
	For link $(i,j) \in \mathcal{\bar{L}}$, if $(i,j) \in \mathcal{L}$, then its link capacity and delay are the same as those in the physical network, i.e., $C_{i,j}$ and $d_{i,j}$; otherwise, we let $ d_{i,j} = 0 $ and $C_{i,j} = +\infty$.
	In Fig. \ref{map_example}, we illustrate the constructed virtual network based on the physical network in Fig. \ref{originalmap} (Recall that in the physical network in Fig. \ref{originalmap}, node E can process functions $f^1$ and $f^2$ while node C can process only function $f^1$). 
	
	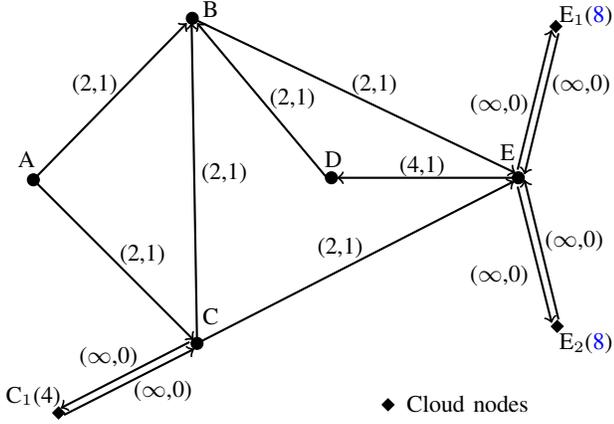
\begin{figure}[h]
		\centering
		\begin{tikzpicture}[scale=\textwidth/20cm]
			\begin{small}
				\draw [->,line width=0.8pt] (1.6,5.4) -- (3.85,7.7);
				\draw [->,line width=0.8pt] (1.6,5.4) -- (3.95,3.05);
				\draw [->,line width=0.8pt] (3.93,7.77) -- (8.67,5.5);
				\draw [->,line width=0.8pt] (4,3) -- (8.65,5.38);
				\draw [->,line width=0.8pt] (4,3) -- (3.92,7.7);
				\draw [->,line width=0.8pt] (8.71,5.43) -- (6.05,5.43);
				\draw [->,line width=0.8pt] (5.97,5.35) -- (4,7.7);
				\draw [->,line width=0.8pt] (3.89,3.02) -- (2,2.05);
				\draw [->,line width=0.8pt] (2.05,1.95) -- (3.96,2.92);
				\draw [->,line width=0.8pt] (8.7,5.55) -- (9.2,7.6);
				\draw [->,line width=0.8pt] (9.3,7.55) -- (8.8,5.48); 
				\draw [->,line width=0.8pt] (8.7,5.3) -- (9.2,3.3);
				\draw [->,line width=0.8pt] (9.3,3.35) -- (8.8,5.4);
				\draw [fill=black] (1.6,5.4) circle (2.5pt);
				\draw[color=black] (1.5,5.7) node {A};
				\draw [fill=black] (3.93,7.77) circle (2.5pt);
				\draw[color=black] (4.2,7.9) node {B};
				\draw [fill=black] (4,3) circle (2.5pt);
				\draw[color=black] (4.2,3.4) node {C};
				\draw [fill=black] (5.97,5.43) circle (2.5pt);
				\draw[color=black] (6,5.7) node {D};
				\draw [fill=black] (8.71,5.43) circle (2.5pt);
				\draw[color=black] (8.55,5.8) node {E};
				\draw[color=black] (2.5,6.8) node {(2,1)};
				\draw[color=black] (3.2,4.3) node {(2,1)};
				\draw[color=black] (6.6,6.8) node {(2,1)};
				\draw[color=black] (6.1,4.4) node {(2,1)};
				\draw[color=black] (4.4,5.5) node {(2,1)};
				\draw[color=black] (7.3,5.6) node {(4,1)};
				\draw[color=black] (5.4,6.6) node {(2,1)};
				\draw (6.95,2.36) node[anchor=north west] {{Cloud nodes}};
				\draw [fill=black] (6.8,2.1) ++(-2.5pt,0 pt) -- ++(2.5pt,2.5pt)--++(2.5pt,-2.5pt)--++(-2.5pt,-2.5pt)--++(-2.5pt,2.5pt);
				\draw [fill=black] (9.28,3.25) ++(-2.5pt,0 pt) -- ++(2.5pt,2.5pt)--++(2.5pt,-2.5pt)--++(-2.5pt,-2.5pt)--++(-2.5pt,2.5pt);
				\draw[color=black] (9.7,3) node {$\text{E}_2$({\color{blue}{8}})};
				\draw [fill=black] (9.26,7.65) ++(-2.5pt,0 pt) -- ++(2.5pt,2.5pt)--++(2.5pt,-2.5pt)--++(-2.5pt,-2.5pt)--++(-2.5pt,2.5pt);
				\draw[color=black] (9.7,7.8) node {$\text{E}_1$({\color{blue}{8}})};
				\draw [fill=black] (1.97,1.98) ++(-2.5pt,0 pt) -- ++(2.5pt,2.5pt)--++(2.5pt,-2.5pt)--++(-2.5pt,-2.5pt)--++(-2.5pt,2.5pt);
				\draw[color=black] (1.6,2.22) node {$\text{C}_1$({{4}})};
				\draw[color=black] (9.53,4.5) node {($ \infty $,0)};
				\draw[color=black] (8.43,4) node {($ \infty $,0)};
				\draw[color=black] (9.63,6.8) node {($ \infty $,0)};
				\draw[color=black] (8.43,6.5) node {($ \infty $,0)};
				\draw[color=black] (3.5,2.3) node {($ \infty $,0)};
				\draw[color=black] (2.7,2.8) node {($ \infty $,0)};
			\end{small}
		\end{tikzpicture}
		\label{transformedmap}
		\caption{The virtual network corresponding to the physical network in Fig. \ref{originalmap}. Notice that in the virtual network, nodes C and E are not cloud nodes any more; the sum of cloud nodes $\text{E}_1$'s and $\text{E}_2$'s capacities should not exceed node E's capacity in the physical network.}
		\label{map_example}
	\end{figure}

	In the constructed virtual network, we can, without loss of generality, require that {\emph{if flow $ k $ goes into some virtual cloud node $ v_t $, exactly one service function of flow $ k $'s SFC must be processed by it.}}
	Let us consider two cases. 
	The first case is that flow $k$ passes through cloud node $v$ without any service
	function being processed in the physical network, then, in the virtual network flow
$ k $ does not go into cloud nodes $ v_1, \ldots, v_{m_v} $ (but still goes into node $v$).
	The second case is that flow $ k $ passes through cloud node $ v $ with
$ \tau~(1 \leq \tau \leq m_v) $ service functions being processed in the physical network,
	then, in the virtual network flow $k$ can go into $ \tau $ of the cloud nodes
$ v_1,\ldots, v_{m_v} $ and each of them process only one service function for
	flow $k$.
	In a nutshell, although at most one function of each flow can be processed by a
	cloud node in the virtual network, (possibly) multiple functions of the same flow can
	be processed by the corresponding cloud node in the physical network,
	which plays a critical role in reducing the number of nodes that need to be activated
	and decrease the total communication delay in the physical network.\vspace{0.1cm}\\
	\begin{table}[t]
		\label{notation}
		\caption{Summary of Notations in problem \eqref{mip}.}
		\centering	
		\setlength{\tabcolsep}{3pt} 
		\renewcommand{\arraystretch}{1.2} 
		\begin{tabularx}{0.48\textwidth}{|c|X|}
			\hline
			\multicolumn{2}{|c|}{Parameters}   \\
			\hline 
			$ \mu_v $                        & computational capacity of (physical) cloud node $ v $                                                                                                                                                                                \\
			\hline 
			\textchange{$N(v)$}              & \textchange{the only node that links to cloud node $v$ in the virtual network}                                                                                                                                                       \\
			\hline 
			$ C_{i,j} $                      & communication capacity of link $ (i,j) $                                                                                                                                                                                             \\
			\hline 
			$ d_{i,j} $                      & communication delay of link $ (i,j) $                                                                                                                                                                                                \\
			\hline 
			$ \mathcal{F}(k) $               & the index set that corresponds to flow $ k $'s SFC                                                                                                                                                                                   \\
			\hline
			$ f_s^k $                        & the $ s $-th function in the SFC of flow $ k $                                                                                                                                                                                       \\
			\hline 
			$ d_{v,s}(k) $                   & NFV delay that function $f_s^k$ is hosted at cloud node $ v $                                                                                                                                                                        \\
			\hline 
			$ \lambda_s(k) $                 & service function rate after receiving function $ f_s^k $                                                                                                                                                                             \\
			\hline
			$ \Theta_k $                     & E2E latency threshold of flow $ k $                                                                                                                                                                                                  \\
			\hline
			\multicolumn{2}{|c|}{Variables}                                                                                                                                                                                                                                       \\
			\hline 
			$ y_v $                          & binary variable indicating whether or not (physical) cloud node $ v $ is activated                                                                                                                                                   \\
			\hline 
			$ x_{v,s}(k) $                   & binary variable indicating whether or not (virtual) cloud node $ v $ processes function $ f_s^k $                                                                                                                                    \\
			\hline 
			$ x^0_{v,s}(k) $                 & binary variable indicating whether or not (physical) cloud node $ v $ processes function $ f_s^k $                                                                                                                                   \\
			\hline  
			$ r(k,s,v_s, v_{s+1}, p) $       & data rate on the $p$-th path of flow $(k,s,v_s,v_{s+1})$  that is used to route the traffic flow from (virtual) cloud node $ v_s $ to (virtual) cloud node $ v_{s+1} $ (hosting functions $ f_s^k $ and $ f_{s+1}^k $, respectively) \\
			\hline  
			$ z_{i,j}(k,s,v_s, v_{s+1}, p) $ & binary variable indicating whether or not
			link $ (i,j) $ is on the $p$-th path of flow $(k,s,v_s,v_{s+1})$   \\
			\hline  
			$ r_{i,j}(k,s,v_s, v_{s+1}, p) $ & data rate on link $ (i,j) $, which is used by the $ p $-th path of flow $(k,s,v_s,v_{s+1})$                                                                                                                                          \\
			\hline 
			$ \theta(k,s) $                  & communication delay due to the traffic flow from the cloud node hosting function $ f_s^k $ to the cloud node hosting function $ f_{s+1}^k $                                                                                          \\
			\hline 
			$ \theta_L(k) $                  & total communication delay of flow $ k $                                                                                                                                                                                              \\
			\hline 
			$ \theta_N(k) $                  & total NFV delay of flow $ k $                                                                                                                                                                                                        \\
			\hline 
		\end{tabularx}
	\end{table}
	{\bf\noindent$\bullet$ VNF Placement and Node Capacity Constraints\vspace{0.1cm}}\\
	\indent We introduce the binary variable $x_{v,s}(k),~s=1,\ldots,\ell_k$, to indicate whether or not node $v$ {in $\mathcal{\bar{V}}$ processes function $f^k_s$} in the virtual network, i.e.,
	\begin{eqnarray*}
		x_{v,s}(k)&=&\left\{\begin{array}{ll}1,
			   & {\text{if~node}}~ v ~{\text{{processes}~function}}~f^k_s;~ \\
			0, & {\text{otherwise}}.\end{array}\right.
	\end{eqnarray*}
	Notice that in practice, node $ v $ may not be able to process function $ f_s^k $ \cite{Zhang2017,Baumgartner2015,Oljira2017}, and in this case, we can simply set $ x_{v,s}(k) = 0 $.
	As analyzed before, in the virtual network, each (virtual) cloud node can process at
	most one service function for each flow:
	\begin{equation}\label{key}
		\sum_{s \in \mathcal{F}(k)} x_{v,s} (k) \leq 1, \ \forall~v \in \mathcal{\barr{V}},\ \forall~k\in \mathcal{K}.
	\end{equation}

	{For notational convenience, we introduce a binary variable $ x^0_{v, s}(k) $ denoting whether or not node $v$ processes function $f^k_s$ in the physical network.
	For each flow $ k $, we require that each service function in its SFC is processed by exactly one cloud node (in the physical network), i.e.,
	\begin{equation}
		\label{onlyonenode}
		\sum_{v\in \mathcal{{V}}}x^0_{v,s}(k)=1,~\forall ~k \in \mathcal{K}, ~\forall ~s\in  \mathcal{F}(k).
	\end{equation}
	By the definitions of $x^0_{v,s}(k)$ and $ x_{v,s}(k) $, we have
	\begin{align}
		  & x^0_{v,s}(k)=x_{v_1,s}(k)+\cdots+x_{v_{m_v},s}(k), \nonumber                                                       \\
		  & \qquad \qquad\qquad\forall~v \in \mathcal{V}, ~\forall~k \in \mathcal{K}, ~  \forall~s\in \mathcal{F}(k)\label{4}.
	\end{align}
	Here $v_1, \ldots, v_{m_v}$ are the cloud nodes in the virtual network corresponding to the cloud node $v$ in the physical network.
	Notice that since $x^0_{v,s}(k) \in \{0,1\}$, the right-hand side of \eqref{4} must be less than or equal to one, which is enforced by the requirement that each function is processed by exactly one cloud node (in the virtual network).}
	
	Let $y_v\in\{0,1\}$ represent the activation of cloud node $v$ (in the physical network), i.e., if $ y_v =1 $, node $ v $ is activated and powered on; otherwise, it is powered off. Thus
	\begin{equation}
		\label{xyrelation}
		\xhatzero_{v,s}(k) \leq  y_v, ~ \forall~v \in \mathcal{V},~\forall~k \in \mathcal{K},~\forall~s \in \mathcal{F}(k). 
	\end{equation}

	Since processing one unit of data rate consumes one unit of (normalized) computational capacity, we can get the node capacity constraints as follows:
	\begin{equation}
		\label{nodecapcons}
		\sum_{k\in \mathcal{K}}\sum_{s \in \mathcal{F}(k)}\lambda_s(k)\xhatzero_{v,s}(k)\leq \mu_v y_v,~\forall~ v \in \mathcal{V}.
	\end{equation}

	The total capacities of activated (physical) cloud nodes should be larger than or equal to the total required data rates of all services. Then, we have
	\begin{equation}
		\label{totalcapacitiescons}
		\sum_{v \in \mathcal{V}} \mu_v y_v \geq \sum_{k\in \mathcal{K}}\sum_{s \in \mathcal{F}(k)}\lambda_s(k) .
	\end{equation}
	Constraint \eqref{totalcapacitiescons} is redundant since it can be obtained by adding all the constraints in \eqref{nodecapcons} and using \eqref{onlyonenode}.
	However, adding such constraint in the problem formulation can potentially improve its solution efficiency. {Similar trick} can be found in \cite{Gadegaard2018} and \cite{Guo2019}.\vspace{0.1cm}\\
	{\bf\noindent$\bullet$ Flexible Routing and Link Capacity Constraints\vspace{0.1cm}}\\
	\indent In practice, each flow $ k $ should go into the cloud nodes in the prespecified order of the functions in its SFC, starting from the source node $ S(k) $ and ending at the destination node $ D(k) $.
	{We use $ (k,s,v_s,v_{s+1}) $ to denote the flow that is routed between (virtual) cloud nodes $ v_s $ and $ v_{s+1} $ hosting functions $ f_s^k $ and $ f_{s+1}^k $, respectively.}
	\textchange{This means that $v_s$ and $v_{s+1}$ are the source and destination nodes of flow $(k,s,v_s,v_{s+1})$.}
	Particularly, if $ s=0 $ and $ s=\ell_k $, we assume without loss of generality that the virtual service functions $ f_0^k $ and $ f_{\ell_k+1}^k $ are hosted at source and destination nodes $ S(k) $ and $ D(k) $, respectively.
	Suppose that there are at most $P$ paths that can be used for routing flow $ (k,s,v_s,v_{s+1}) $.
	In general, such an assumption on the number of paths may affect the solution's quality.
	Indeed, the choice of  $P$ offers a tradeoff between the flexibility of traffic routing in the problem formulation and the computational complexity of solving it: the larger the parameter $P$ is, the more flexibility of routing and the higher the computational complexity.
	A special choice is $P=|\mathcal{\bar{L}}|$.
	Such a choice will not affect the solution's quality.
	In fact, it is a well-known result from classical network flow theory that any routes between two nodes can be decomposed into the sum of at most $ |\mathcal{\barr{L}}| $ routes on the paths and a circulation; see \cite[Theorem 3.5]{Ahuja1993}.
	\textchange{It is also worthwhile mentioning that in practice, setting $P$ to be a small value (e.g., $P=2$) can significantly improve the network performance, as compared with setting $P=1$ \cite{He2008}.}
	
	Denote $ \mathcal{P}=\{1, \ldots,P\} $.
	For flow $ (k,s,v_s,v_{s+1}) $, let $ r(k,s,v_s,v_{s+1},p) $ be the data rate on the $ p $-th path.
	We need to introduce this variable in our formulation, as the traffic flow of each service in our formulation is allowed to transmit on (possibly) multiple paths in order to exploit the flexibility of traffic routing, which is in sharp contrast to the formulation in \cite{Woldeyohannes2018}.
	As can be seen later (e.g., in Eqs. \eqref{relalambdaandx1}-\eqref{relalambdaandx3}, \eqref{mediacons1}, \eqref{mediacons3}, \eqref{firstcons1}, \eqref{firstcons3}, \eqref{lastcons1}, and \eqref{lastcons3}), this variable plays an important role in the flow conversation constraints associated with the $p$-th path and cloud nodes $v_s$ and $v_{s+1}$.
	Notice that by \eqref{key} and the fact that $ S(k), D(k) \notin \mathcal{V} $, we must have $ v_s \neq v_{s+1} $.
	For each $ k \in \mathcal{K} $, from the definitions of $ x_{v_s,s}(k) $, $ x_{v_{s+1}, s+1}(k) $, and $ r(k,s,v_s, v_{s+1}, p) $, we have
	\begingroup
	\allowdisplaybreaks
	\begin{align} 
		  & \sum_{p \in \mathcal{P}}  r(k, {s}, v_s, v_{s+1}, p) =  \lambda_{s}(k) x_{v_s,s}(k)  x_{v_{s+1},{s+1}}(k),  \nonumber                             \\ 
		  & \qquad \qquad\qquad \qquad  \forall~s\in \mathcal{F}(k)\backslash\{\ell_k\},~\forall~v_s, v_{s+1} \in \mathcal{\barr{V}}, \label{relalambdaandx1} \\
		  & \sum_{p \in \mathcal{P}}  r(k, 0, S(k),v_1, p)=\lambda_{0}(k) x_{v_1,1}(k),  ~\forall~v_1 \in \mathcal{\barr{V}}, \label{relalambdaandx2}         \\
		  & \sum_{p \in \mathcal{P}}  r(k, {\ell_k}, v_{\ell_k},D(k), p) =\lambda_{\ell_k}(k)x_{v_{\ell_k},{\ell_k}}(k),  \nonumber                           \\ 
		  & \qquad \qquad\qquad \qquad\qquad\qquad\qquad\qquad\qquad
		\forall~v_{\ell_k} \in \mathcal{\barr{V}}. \label{relalambdaandx3}
	\end{align}
	\endgroup
	Constraint \eqref{relalambdaandx1} indicates that if the $ s $-th and $ (s+1) $-th functions of flow $ k $ (i.e., functions $ f_s^k $ and $ f_{s+1}^k $) are hosted at (virtual) cloud nodes $ v_s $ and $ v_{s+1} $, respectively, then the total {data} rates sent from $ v_s $ to $ v_{s+1} $ must be equal to $ \lambda_s(k) $.
	Similarly, if function $ f_1^k $ is hosted at (virtual) cloud node $ v_1 $, constraint \eqref{relalambdaandx2} guarantees that the total data rates sent from $S(k) $ to $ v_{1} $ must be equal to $ \lambda_0(k) $; if function $ f_{\ell_k}^k $ is hosted at (virtual) cloud node $ v_{\ell_k} $, constraint \eqref{relalambdaandx3} guarantees that total data rates sent from $v_{\ell_k} $ to $ D(k) $ must be equal to $ \lambda_{\ell_k}(k) $.

	We then use $ z_{i,j}(k, s,v_s, v_{s+1}, p) =1 $ to denote that link $ (i,j) $ is on the $ p $-th path of flow $ (k,s,v_s,v_{s+1}) $; otherwise, $ z_{i,j}(k, s,v_s, v_{s+1}, p) =0 $.
	By definition, for all $ k \in \mathcal{K} $, $ p \in \mathcal{P} $, and $ (i,j) \in \mathcal{\barr{L}} $, we have
	\begingroup
	\allowdisplaybreaks
	\begin{align}
		  & z_{i,j}(k, s, v_s,v_{s+1}, p ) \leq  x_{v_s,s}(k) x_{v_{s+1}, {s+1}}(k),\nonumber                                                             \\
		  & \qquad \qquad\qquad\qquad \forall~s \in \mathcal{F}(k)\backslash\{\ell_k\},  ~\forall~v_s, v_{s+1} \in \mathcal{\barr{V}}, \label{relazandx1} \\
		  & z_{i,j}(k, 0, S(k), v_{1}, p ) \leq  x_{v_1,1}(k), ~ \forall~v_1 \in  \mathcal{\barr{V}}, \label{relazandx2}                                  \\		
		  & z_{i,j}(k, {\ell_k}, v_{\ell_k}, D(k),p ) \leq  x_{v_{\ell_k},{\ell_k}}(k),~\forall~v_{\ell_k }\in  \mathcal{\barr{V}}. \label{relazandx3}
	\end{align}
	\endgroup
	
	Recall that, in the virtual network, if flow $ k $ goes into cloud node $ v $, exactly one service function in flow $ k $'s SFC must be processed by this node.
	Then, if $v = v_s$ or $ v= v_{s+1} $, at most one of cloud node $v$'s two links can be used by the $ p $-th path of flow $ (k,s,v_s, v_{s+1}) $; otherwise none of cloud node $v$'s two links can be used by the $ p $-th path of flow $ (k,s,v_s, v_{s+1}) $.
	Therefore, for each $ v \in \mathcal{\barr{V}} $, $ k \in \mathcal{K} $, $s \in \mathcal{F}(k)\cup \{0\}$, $v_s, v_{s+1} \in \mathcal{\bar{V}}$, and $ p \in \mathcal{P} $, we have
	\begin{equation*}
		\textchange{z_{v,{N}(v)}(k,s,v_s,v_{s+1},p) + z_{N(v),v}(k,s,v_s,v_{s+1},p)}
	\end{equation*}
	\begin{numcases}{}
		\leq 1,     & \text{if}~$v = v_s~\text{or}~v=v_{s+1};$ \label{loops-1} \\
		=0,                                   & otherwise, \label{loops-2}
	\end{numcases}
	{where $N(v)$ denotes the only node that links to cloud node $v$ in the virtual network (and $N(v)$ is also the physical node connecting with virtual cloud node $v$).}

	If $ z_{i,j}(k, s,v_s, v_{s+1}, p) =1 $, let $ r_{i,j}(k, s,v_s, v_{s+1}, p )  $ denote the associated amount of {data} rate.
	By definition, for each $ (i,j) \in \mathcal{\barr{L}} $, $ k \in \mathcal{K} $, $s \in \mathcal{F}(k)\cup \{0\}$, $v_s, v_{s+1} \mathcal{\in \barr{V}}$, and $ p\in \mathcal{P} $, we have the following coupling constraint:
	\begin{align}
		  & r_{i,j}(k, s, v_s, v_{s+1}, p ) \leq \lambda_{s}(k)  z_{i,j}(k, s, v_s, v_{s+1},p ).         \label{relarandz1}
	\end{align}
	The total data rates on link $ (i,j) $ is upper bounded by capacity $ C_{i,j} $:
	\begin{align}
		  & \label{linkcapcons}
		\sum_{k \in \mathcal{K}} \sum_{s\in \mathcal{F}(k) \cup \{0\}} \sum_{v_s,v_{s+1}\in \mathcal{\barr{V}}}   \sum_{p \in \mathcal{P}} r_{i,j}(k, s, v_s, v_{s+1}, p)  \nonumber \\
		  & \qquad \qquad\qquad\qquad\qquad\qquad\leq C_{i,j},  ~\forall~(i,j) \in \mathcal{L} .
	\end{align}
	{\bf\noindent$\bullet$ SFC Constraints\vspace{0.1cm}}
	
	To ensure the functions of each flow are followed in the prespecified order as in \eqref{sequence}, we need to introduce several constraints below.
	We start with the flow conservation constraints of each intermediate function of each flow.
	In particular, for each $ k \in \mathcal{K} $, $ s \in \mathcal{F}(k)\backslash\{\ell_k\} $, $ v_s, v_{s+1}  \in \mathcal{\barr{V}}$, $ p \in \mathcal{P} $, and $ i \in \mathcal{\barr{I}} $, we have
	\begin{equation*}
		\sum_{j: (j,i) \in \mathcal{\barr{L}}} r_{j,i}(k, s, v_s, v_{s+1}, p) - \sum_{j: (i,j) \in \mathcal{\barr{L}}} r_{i,j}(k, s, v_s, v_{s+1}, p)
	\end{equation*}
	\begin{numcases}{=}
		-r(k, s, v_s, v_{s+1},p),     & ~~~\text{if}~$i = v_s;$ \label{mediacons1} \\
		0,                                   & ~~~\text{if}~$i \neq v_s, ~v_{s+1};$  \label{mediacons2}   \\
		r(k, s, v_s, v_{s+1},p),    &  ~~~\text{if}~$i=v_{s+1};$  \label{mediacons3}
	\end{numcases}
	\begin{equation*}
		\sum_{j: (j,i) \in \mathcal{\barr{L}}} z_{j,i}(k, s, v_s, v_{s+1}, p) - \sum_{j: (i,j) \in \mathcal{\barr{L}}} z_{i,j}(k, s, v_s, v_{s+1}, p)
	\end{equation*}
	\begin{numcases}{=}
		-x_{v_s,s}(k)  x_{v_{s+1},{s+1}}(k),     & \text{if}~$i = v_s;$\label{mediacons4}   \\
		0,                                   & \text{if}~$i \neq v_s, ~v_{s+1};$ \label{mediacons5}  \\
		x_{v_s,s}(k)  x_{v_{s+1},{s+1}}(k)    &  \text{if}~$i=v_{s+1}$.    \label{mediacons6}
	\end{numcases}
	First, note that constraints \eqref{mediacons1}, \eqref{mediacons2}, and \eqref{mediacons3} are flow conservation constraints for the data rate.
	Second, we need another three flow conservation constraints \eqref{mediacons4}, \eqref{mediacons5}, and \eqref{mediacons6}.
	To be more precise, for each pair of cloud nodes $ v_s $ and $ v_{s+1} $, considering constraints \eqref{mediacons4}, \eqref{mediacons5}, and \eqref{mediacons6}, we only need to look at the case that $ x_{v_s,s}(k) =1 $ and $ x_{v_{s+1}, s+1}(k)=1 $, since otherwise from constraint \eqref{relazandx1}, all the variables $ z_{i,j}(k,s,v_s,v_{s+1}, p) $ in \eqref{mediacons4}, \eqref{mediacons5}, and \eqref{mediacons6} must be equal to zero.
	Constraint \eqref{mediacons5} enforces that for every node that does not host
	functions $ f_s^k $ and $ f_{s+1}^k $, if one of its incoming links is assigned for the
$ p $-th path of flow $(k,s,v_s,v_{s+1})$, then one of its outgoing links must
	also be assigned for this path.
	Similarly, constraint \eqref{mediacons4} implies that, if function $ f_s^k $ is hosted
	at (virtual) cloud node $ v_s $, node $ v_s $'s outgoing link must be assigned for the $ p $-th path
	of flow $(k,s,v_s,v_{s+1})$ and node $ v_s $'s incoming link cannot be assigned
	for this path; constraint \eqref{mediacons6} implies that if function $ f_{s+1}^k $ is
	hosted at (virtual) cloud node $ v_{s+1} $, node $ v_{s+1} $'s outgoing link cannot be assigned for
	the $ p $-th path of flow $(k,s,v_s,v_{s+1})$ and node $ v_{s+1} $'s
	incoming link must be assigned for this path\footnote{Notice that every cloud node in the
		virtual network
		has only a single outgoing link and a single incoming link.}.
	{In our formulation, all nodes including the inactivated cloud nodes can be used for routing the traffic flow, i.e., they can serve as intermediate nodes in the associated paths, as implied by constraints \eqref{mediacons1}-\eqref{mediacons6}.}
	
	We next present the flow conservation constraints of the first function of each flow. For all $ k \in \mathcal{K} $, $ v_1  \in \mathcal{\barr{V}}$, $ p \in \mathcal{P} $, and $ i \in \mathcal{\barr{I}} $, similar to constraints \eqref{mediacons1}-\eqref{mediacons6}, we have
	\begin{equation*}
		\sum_{j: (j,i) \in \mathcal{\barr{L}}} r_{j,i}(k, 0, S(k), v_{1}, p) - \sum_{j: (i,j) \in \mathcal{\barr{L}}} r_{i,j}(k, 0,  S(k), v_{1}, p) 
	\end{equation*}
	\begin{numcases}{=}
		-r(k, 0, S(k),v_1, p),     & \text{if}~$i = S(k);$ \label{firstcons1} \\
		0,                                   & \text{if}~$i \neq S(k), ~v_{1};$  \label{firstcons2}   \\
		r(k, 0, S(k),v_1, p),   &  \text{if}~$i=v_{1};$  \label{firstcons3}
	\end{numcases}
	\begin{equation*}
		\sum_{j: (j,i) \in \mathcal{\barr{L}}} z_{j,i}(k, 0,  S(k), v_{1},p) - \sum_{j: (i,j) \in \mathcal{\barr{L}}} z_{i,j}(k, 0,  S(k), v_{1}, p)
	\end{equation*}
	\begin{numcases}{=}
		-x_{v_1,1}(k),     & \qquad \qquad\text{if}~$i = S(k);$\label{firstcons4}   \\
		0,                                   &\qquad \qquad \text{if}~$i \neq S(k), ~v_{1};$  \label{firstcons5}  \\
		x_{v_1,1}(k)   &\qquad \qquad  \text{if}~$i=v_{1}$.    \label{firstcons6}
	\end{numcases}

	Finally, we present flow conservation constraints of the last function of each flow. For all $ k \in \mathcal{K} $, $ v_{\ell_k}  \in \mathcal{\barr{V}}$, $ p \in \mathcal{P} $, and $ i \in \mathcal{\barr{I}} $, similar to constraints \eqref{mediacons1}-\eqref{mediacons6}, we have
		{
			\small 
			\begin{equation*}
				\sum_{j: (j,i) \in \mathcal{\barr{L}}} r_{j,i}(k, {\ell_k}, v_{\ell_k}, D(k),p) -\sum_{j: (i,j) \in \mathcal{\barr{L}}} r_{i,j}(k, {\ell_k}, v_{\ell_k}, D(k), p)
			\end{equation*}
			\begin{numcases}{=}
				-r(k, {\ell_k}, v_{\ell_k},D(k),p),     & \text{if}~$i = v_{\ell_k};$ \label{lastcons1} \\
				0,                                   & \text{if}~$i \neq v_{\ell_k}, ~D(k);$  \label{lastcons2}   \\
				r(k, {\ell_k}, v_{\ell_k},D(k),p),    &  \text{if}~$i=D(k);$  \label{lastcons3}
			\end{numcases}
			\begin{equation*}
				\sum_{j: (j,i) \in \mathcal{\barr{L}}} z_{j,i}(k, {\ell_k}, v_{\ell_k},  D(k),p) -\sum_{j: (i,j) \in \mathcal{\barr{L}}} z_{i,j}(k,  {\ell_k}, v_{\ell_k},  D(k), p)
			\end{equation*}
			\begin{numcases}{=}
				-x_{v_{\ell_k}, {\ell_k}}(k),     & \qquad\qquad\text{if}~$i = v_{\ell_k};$\label{lastcons4}   \\
				0,                                   & \qquad\qquad\text{if}~$i \neq v_{\ell_k}, ~D(k);$  \label{lastcons5}  \\
				x_{v_{\ell_k}, {\ell_k}}(k),    &  \qquad\qquad\text{if}~$i=D(k)$.    \label{lastcons6}
			\end{numcases}\vspace{-0.25cm}\\
		}
		{\bf\noindent$\bullet$ E2E Latency Constraints\vspace{0.1cm}}

	Next, we consider the delay constraints of each flow.
	Let $ \theta(k,s) $ be the variable denoting the communication delay due to the traffic flow from the cloud node hosting function $ f^k_s $ to the cloud node hosting function $ f^k_{s+1} $.
	Then, $ \theta(k,s) $ should be the largest one among the $ P $ paths, i.e.,
	\begin{align}
		  & \theta(k,s) \geq \sum_{v_s, v_{s+1}\in\mathcal{\barr{V}}}  \sum_{(i,j) \in \mathcal{L}}  d_{i,j}  z_{i,j}(k, s, v_s, v_{s+1}, p), \nonumber \\
		  & \qquad \qquad\forall~k \in \mathcal{K}, ~ \forall~s \in \mathcal{F}(k) \cup \{0\}, ~\forall ~p \in \mathcal{P}  \label{consdelay2funs}.
	\end{align}
	Hence the total communication delay on the links of flow $ k $, denoted as $ \theta_L(k) $, can be written as
	\begin{equation}
		\label{linkdelaycons}
		\theta_L(k) = \sum_{s \in \mathcal{F}(k)\cup \{0\}} \theta(k,s), ~\forall~ k \in  \mathcal{K}.
	\end{equation}
	Now for each flow $ k $, we consider the total NFV delay on the nodes, denoted as $ \theta_N(k) $.
	This can be written as
	\begin{equation}
		\label{nodedelaycons}
		\theta_N(k) = \sum_{s \in \mathcal{F}(k)} \sum_{v \in \mathcal{{V}}} d_{v,s}(k) \xhatzero_{v,s}(k),~\forall  ~k \in  \mathcal{K}.
	\end{equation}
	The E2E delay of flow $k$ is the sum of total communication delay $ \theta_L(k) $ and total NFV delay $ \theta_N(k) $ \cite{Luizelli2015,Woldeyohannes2018,Qu2019}.
	The following delay constraint ensures that flow $k$'s {E2E} delay is less than or equal to its threshold $\Theta_k$:
	\begin{equation}
		\label{delayconstraint}
		\theta_L(k)+\theta_N(k) \leq \Theta_k,~\forall~k \in  \mathcal{K}.
	\end{equation}
	
	\subsection{A New MBLP Formulation}
	\label{subsec:mipformulation}
	There are two objectives in our problem. 
	The first objective is to minimize the total power consumption of the whole cloud network.
	The power consumption of a cloud node is the combination
	of the dynamic load-dependent power consumption
	(that increases linearly with the load) and the static power
	consumption \cite{3gpp}.
	Hence, the first objective function can be written as:
	\begin{align}
		\sum_{v \in \mathcal{V}}\left[\beta_1y_v+{\Delta} \sum_{k \in \mathcal{K}}\sum_{s \in \mathcal{F}(k)} \lambda_s(k)x^0_{v,s}(k)\right] +\sum_{v \in \mathcal{V}}\beta_2(1-y_v)\label{objfuns}.
	\end{align}
	In the above, the parameters $\beta_1$ and $\beta_2$ are the power consumptions of each activated cloud node and inactivated cloud node, respectively, satisfying $\beta_1>\beta_2$; the parameter $ \Delta $ is the power consumption of processing one unit of data rate.
	From \eqref{onlyonenode}, the above objective function can be simplified as
	\begin{equation*}
		\label{objectivefuns}
		(\beta_1-\beta_2)\sum_{v \in \mathcal{V}}y_v +c, 
	\end{equation*}
	where $  c=  \beta_2|\mathcal{V}|+{\Delta}\sum_{k \in \mathcal{K}}\sum_{s \in \mathcal{F}(k)} \lambda_s(k) $ is a constant.
	Hence, minimizing the total power consumption is equivalent to minimizing the total number of activated cloud nodes.
	{We remark that the total number of activated cloud nodes also approximately reflects the reliability of the whole cloud network.
	Indeed, the latter is the product of the reliability of all activated cloud nodes \cite{Promwongsa2020}, and when the cloud nodes have the same reliability, minimizing the total number of activated cloud nodes is equivalent to maximizing the reliability of the whole cloud network.}

	The second objective is to minimize the total delay of all the services:
	\begin{equation}\label{delobj}
		\sum_{k \in \mathcal{K}} (\theta_L(k) + \theta_N(k)).
	\end{equation}
	{The second objective is important in the following sense.
	First, it will help to avoid cycles in each path connecting the two adjacent service functions in the solution of the problem.
	Second, the problem of minimizing the total number of activated cloud nodes often has multiple solutions and the total E2E delay term can be regarded as a regularizer to make the problem have a unique solution as observed in our simulation results.
	Third, for some delay critical tasks (e.g., maximizing the freshness of information \cite{Kaul2012} in the monitoring system), the total E2E delay is expected to be as small as possible.}

	The above two objectives
	can be combined into a single objective, using the traditional weighted sum method \cite{Marler2010}.
	Based on the above analysis, we present the problem formulation to minimize a weighted sum {of} the total power consumption of the whole cloud network (equivalent to the total number of activated cloud nodes in the physical network) and the total delay of all services:
	\begin{align}
		  & \min_{\boldsymbol{x},\boldsymbol{y},\boldsymbol{z},\boldsymbol{r},\boldsymbol{\theta}} &   & \sum_{v \in \mathcal{V}}y_v + \sigma \sum_{k \in \mathcal{K}} (\theta_L(k) + \theta_N(k)) \nonumber \\
		  & ~~~~~{\text{s.t.~}}                                                                    &   & (\ref{key})-(\ref{delayconstraint}), \label{mip}
		\tag{\rm{NS-I}}
	\end{align}
	where $ \sigma $ is a constant number that balances the importance of the two terms in the objective function.
	\textchange{Before discussing the analysis results of problem \eqref{mip}, we note that our proposed formulation \eqref{mip} can take the node/link availability \cite{Vassilaras2017} into consideration.
		Indeed, if a specific set of nodes/links is not available, due for instance to network issues (security, power, etc), we can set the associated computational/communication capacities to be zero and solve the corresponding problem again.
	}

	We now present some analysis results of problem \eqref{mip}. First, problem \eqref{mip} is an MBLP since the nonlinear terms of binary variables $  x_{v_s,s}(k)  x_{v_{s+1},{s+1}}(k)  $ in \eqref{relalambdaandx1}, \eqref{relazandx1}, \eqref{mediacons4}, and \eqref{mediacons6} can be equivalently linearized \cite{Conforti2014}.
	To be more precise, we can equivalently replace the term $  x_{v_s,s}(k)  x_{v_{s+1},{s+1}}(k)  $ by an auxiliary binary variable $ w_{v_s, v_{s+1},s}(k) $ and add the following constraints:
	\begin{align*}
		  & w_{v_s, v_{s+1},s}(k)  \leq x_{v_s,s}(k), \ w_{v_s, v_{s+1},s}(k)  \leq x_{v_{s+1},{s+1}}(k), \\
		  & w_{v_s, v_{s+1},s}(k)\geq  x_{v_s,s}(k) + x_{v_{s+1},{s+1}}(k)- 1.
	\end{align*}
	Note that the linearity of all variables in problem \eqref{mip} is vital, which allows to leverage the efficient integer programming solver such as Gurobi \cite{Gurobi} to solve the problem to global optimality.
	
	Second, we can show that problem \eqref{mip} is strongly NP-hard.
	\begin{proposition}
		\label{NPhard}
		Problem \eqref{mip} is strongly {{NP}}-hard. 
	\end{proposition}
	\noindent In fact, problem \eqref{mip} includes the problem in \cite{Zhang2017} as a special case, which does not consider the E2E latency constraints of all services. Since  the problem in \cite{Zhang2017} is strongly NP-hard, it follows that, problem \eqref{mip} is also strongly NP-hard.
	In addition, it is simple to check that both numbers of variables and constraints in problem \eqref{mip} are 
$ \mathcal{O}(|\mathcal{\bar{V}}|^2|\bar{\mathcal{L}}||\mathcal{P}|\sum_{k \in \mathcal{K}}\ell_k) $.
	The strong NP-hardness of problem \eqref{mip} and the huge number of variables and constraints in it make the above approach can only solve the problem associated with small size networks. 
	In the next section, we shall propose an equivalent formulation with a significantly smaller number of variables and constraints.

	{Third, in problem \eqref{mip}, if the power consumption term, or equivalently, the
	total number of activated cloud nodes term, in the objective function is more
	important than the total delay term (where the second term just serves as a
	regularizer), the problem reduces to the two-stage formulation, where the
	first stage minimizes the total power consumption and with the minimum power
	consumption the second stage minimizes the total delay of all services}.
	Using a similar argument as in  \cite[Proposition 1]{Guo2019}, we can show that
	solving problem \eqref{mip} with an appropriate parameter $\sigma$ is equivalent to
	solving the above two-stage formulation.
	\begin{proposition}
		\label{oneproblemcondition}
		Suppose $ \Theta_k > 0 $ for some $ k \in \mathcal{K}$. Then, problem \eqref{mip} with $\sigma\in (0 , 1/{{\sum_{k \in \mathcal{K}}\Theta_k }})$ is equivalent
		to the two-stage problem where the first stage minimizes the total power consumption and with the minimum power consumption the second stage minimizes the total \textchange{E2E} delay.
	\end{proposition}
	
	Finally, it is worthwhile highlighting the connection and difference between our proposed formulation \eqref{mip} and that in \cite{Woldeyohannes2018}.
	If we set $P=1$ in \eqref{mip}, then our formulation reduces to that in \cite{Woldeyohannes2018}.
	In particular, the variables $ \{r_{i,j}(k,s,v_s,v_{s+1},p)\} $ in \eqref{linkcapcons} can be replaced by the right-hand sides of \eqref{relarandz1} and all constraints related to the variables $\{r(k,s,v_s,v_{s+1},p)\}$ (e.g., \eqref{relalambdaandx1}-\eqref{relalambdaandx3}, \eqref{relarandz1}, \eqref{mediacons1}-\eqref{mediacons3}, \eqref{firstcons1}-\eqref{firstcons3}, \eqref{lastcons1}-\eqref{lastcons3}) can be removed.
	Our proposed formulation with $P>1$ allows the traffic flows to transmit over (possibly) multiple paths and hence fully exploits the flexibility of traffic routing.
	
	\section{A Compact Problem Formulation}
	\label{compactformulation}
	
	In this section, we shall derive a compact problem formulation for the network slicing problem with a significantly smaller number of variables and constraints. We shall show that this new formulation is indeed equivalent to formulation \eqref{mip}.
	
	\begin{table}[h]
		\label{notation2}
		\caption{Summary of New Variables in Problem \eqref{newmip}.}
		\centering	
		\setlength{\tabcolsep}{4pt} 
		\renewcommand{\arraystretch}{1.2} 
		\begin{tabularx}{0.48\textwidth}{|c|X|}
			\hline  
			$ r(k,s, p) $      & data rate on the $p$-th path of flow $(k,s)$  that is used to route the traffic flow between the two (virtual) cloud nodes hosting functions $ f_s^k $ and $ f_{s+1}^k $, respectively \\
			\hline  
			$ z_{i,j}(k,s,p) $ & binary variable indicating whether or not link $ (i,j) $ is on the $p$-th path of flow $(k,s)$                                                                                         \\
			\hline  
			$ r_{i,j}(k,s,p) $ & data rate on link $ (i,j) $ which is used by the $ p $-th path of flow $(k,s)$                                                                                                         \\
			\hline 
		\end{tabularx}
	\end{table}
	\subsection{A New Problem Formulation}
	{\bf\noindent$\bullet$ Key New Notations and Related Constraints\vspace{0.1cm}}\\
	\indent In the new formulation, we use the same placement variables as in problem \eqref{mip}, i.e., $ x_{v,s}(k) $, $ x^0_{v,s}(k) $, and $ y_v $.
	Hence, we also enforce the same constraints \eqref{key}-\eqref{totalcapacitiescons} in the new formulation.
	In addition, the same delay variables $\theta(k,s)$, $\theta_N(k)$,  and $\theta_L(k)$ are also used.
	As a result, constraints \eqref{linkdelaycons}-\eqref{delayconstraint} are also enforced in the new formulation.
	However, for each flow $ k $, we use different variables to represent its traffic flows.
	Next, we shall discuss the new variables to represent the flows and the related constraints.

	{Recall that in the previous section, we use $(k,s,v_s,v_{s+1})$
	to denote the flow that is routed between source node $v_s$ and destination node $v_{s+1}$ hosting two adjacent functions $f_s^k$ and $f_{s+1}^k$, respectively.
	This makes it easier to present the flow conservation constraints \eqref{mediacons1}-\eqref{lastcons6} as the source and destination nodes (i.e., $v_s$ and $v_{s+1}$, respectively) of a traffic flow are explicitly specified.
	However, notation $(k,s,v_s,v_{s+1})$ can be simplified since the source and destination nodes of a traffic flow can be implicitly represented by variables $\{x_{v}(k,s)\}$.
	Indeed, we can use $(k,s)$ to denote the flow that is routed between the two (virtual) cloud nodes hosting two adjacent functions $ f_s^k $ and $ f_{s+1}^k $, respectively.
	Then, for cloud node $v$, $x_{v}(k,s)=1$ implies that $v$ is the source node and the destination node of flows $(k,s)$ and $(k,s-1)$, respectively.
	Below we shall present the related variables and constraints based on this new notation $(k,s)$.}

	Similar to formulation \eqref{mip}, we assume that there are at most $P$ paths that can be used to route flow $(k,s)$. Let $ r(k,s,p) $ be the data rate on the $ p $-th path of flow $ (k,s) $.
	Then, analogous to constraints \eqref{relalambdaandx1}-\eqref{relalambdaandx3} that enforce the total data rates between the two nodes hosting functions $ f_s^k $ and  $ f_{s+1}^k $ to be equal to $ \lambda_s(k) $, we have
	\begin{align}
		  & \sum_{p \in \mathcal{P}}  r(k, {s}, p) =  \lambda_{s}(k),   ~ \forall~ k \in \mathcal{K},~\forall~s\in \mathcal{F}(k)\cup \{0\} \label{relalambdaandx11}.                                          
	\end{align}

	Let $ z_{i,j}(k,s,p)=1 $ denote that link $ (i,j) $ is on the $ p $-th path of flow $ (k,s) $; otherwise $  z_{i,j}(k,s,p)=0 $. Similarly in constraints \eqref{loops-1}-\eqref{loops-2}, we need the following constraints to guarantee that at most one of the two links associated with (virtual) cloud node $ v $, i.e., $ (v,\textchange{N(v)}) $ and $ (\textchange{N(v)},v)$,  is used by the $ p $-th path of flow $ (k,s) $:
	\begin{align}
		\label{loops}
		  & z_{v,\textchange{N(v)}}(k,s,p) + z_{\textchange{N(v)},v}(k,s,p) \leq 1, \nonumber                                                   \\
		  & ~~ ~\forall~v \in \mathcal{\bar{V}},~\forall~k \in \mathcal{K}, ~\forall~s \in \mathcal{F}(k)\cup \{0\},~\forall~p \in \mathcal{P}.
	\end{align}
	Moreover, similar to constraint \eqref{consdelay2funs}, we have
	\begin{align}
		  & \theta(k,s) \geq \sum_{(i,j) \in \mathcal{\bar{L}}}  d_{i,j}  z_{i,j}(k, s, p),\nonumber                                                \\
		  & \qquad\qquad  \forall~k \in \mathcal{K}, ~ \forall~s \in \mathcal{F}(k) \cup \{0\},~\forall ~p \in \mathcal{P} \label{consdelay2funs1}.
	\end{align}

	If $ z_{i,j}(k,s,p)=1  $, let $ r_{i,j}(k,s,p) $ be the associated data rate. By definition, we have the following coupling constraints:
	\begin{align}
		  & r_{i,j}(k, s,  p ) \leq \lambda_{s}(k)  z_{i,j}(k, s,p ), \nonumber                                                                                        \\
		  & ~~\forall~(i,j) \in \bar{\mathcal{L}}, ~\forall~k \in \mathcal{K}, ~\forall~s \in \mathcal{F}(k)\cup \{0\},~\forall~p \in \mathcal{P}. \label{relarandz11}
	\end{align}
	The new constraint that enforces the total data rates on link $ (i,j) $ is upper bounded by capacity $ C_{i,j} $
	\begin{equation}
		\label{linkcapcons1}
		\sum_{k \in \mathcal{K}} \sum_{s\in \mathcal{F}(k) \cup \{0\}}\sum_{p \in \mathcal{P}}  r_{i,j}(k, s,p) \leq C_{i,j}, ~  \forall~(i,j) \in \mathcal{L} .
	\end{equation}

	{\bf\noindent$\bullet$ New SFC Constraints\vspace{0.1cm}}\\
	\indent Next, we present the flow conservation constraints in the new formulation to ensure that the functions of each flow are processed in the prespecified order as in \eqref{sequence}.
	We start with the flow conservation constraints of the intermediate functions of each flow.
	In particular, for each $ k \in \mathcal{K} $, $ s \in \mathcal{F}(k)\backslash\{\ell_k\} $, $ p \in \mathcal{P} $, and $ i \in \mathcal{\barr{I}} $, we have\vspace{0.1cm}
	{\small\begin{numcases}{\!\!\!\!\!\!\!\!\!\!}
			\sum_{j: (j,i) \in \mathcal{\barr{L}}} r_{j,i}(k, s, p) - \sum_{j: (i,j) \in \mathcal{\barr{L}}} r_{i,j}(k, s,  p)=0,     & \!\!\!\!\!\!\!\!\!\!\text{if}~$i   \in \bar{\mathcal{I}}\backslash \bar{\mathcal{V}} ;$ \label{mediacons11} \\
			r_{i,\textchange{N(i)}}(k, s,  p)=r(k,s,p)x_{i,s}(k),                                   & \!\!\!\!\!\!\!\!\!\!\text{if}~$ i \in \bar{\mathcal{V}}; $  \label{mediacons12}   \\
			r_{\textchange{N(i)},i}(k, s, p)=r(k, s, p) x_{i,s+1}(k),   &  \!\!\!\!\!\!\!\!\!\!\!\! \text{if}~$i \in \bar{\mathcal{V}};$  \label{mediacons13}
		\end{numcases}
		\vspace{0.1cm}
		\begin{numcases}{\!\!\!\!\!\!\!\!\!\!}
			\sum_{j: (j,i) \in \mathcal{\barr{L}}} z_{j,i}(k, s, p) - \sum_{j: (i,j) \in \mathcal{\barr{L}}} z_{i,j}(k, s,  p)=0,     & \!\!\!\!\!\!\!\!\!\!\text{if}~$i   \in \bar{\mathcal{I}}\backslash \bar{\mathcal{V}} ;$ \label{mediacons14} \\
			z_{i,\textchange{N(i)}}(k, s,  p)= x_{i,s}(k),                                   & \!\!\!\!\!\!\!\!\!\!\text{if}~$ i \in \bar{\mathcal{V}} ;$  \label{mediacons15}   \\
			z_{\textchange{N(i)},i}(k, s, p)= x_{i,s+1}(k),   &  \!\!\!\!\!\!\!\!\!\!\!\! \text{if}~$i \in \bar{\mathcal{V}}.$  \label{mediacons16}
		\end{numcases}}
	Constraint \eqref{mediacons14} enforces that for a (virtual) intermediate node that does not process the function, if one of its incoming links is assigned for the $ p $-th path of flow $(k,s)$, then one of its outgoing links must also be assigned for this path.
	Constraint \eqref{mediacons11} further implies that the data rates over the two links are the same.
	Recall that, in the virtual network, there are exactly two links related to each cloud node
$ i $, i.e., $ (i,\textchange{N(i)}) $ and $ (\textchange{N(i)},i) $ and if flow $ k $ goes into cloud node $ i $,
	exactly one service function in flow $ k $'s SFC must be processed by this node.
	This implies that flow $ (k,s) $ comes out of cloud node  $ i $  using link $ (i,\textchange{N(i)}) $
	if and only if function $ f_s^k $ is hosted at cloud node $ i $, i.e., $ x_{i,s}(k) =1  $.
	This is enforced by constraint \eqref{mediacons15}. 
	In addition, if $ x_{i,s}(k) =1  $, the above also requires that the incoming link of
	cloud node $ i $, i.e., $ (\textchange{N(i)},i) $, cannot be used by the $ p $-th path of flow
$ (k,s) $, which is enforced by constraints \eqref{loops} and \eqref{mediacons15}.
	Furthermore, we need constraint \eqref{mediacons12} to enforce that if
$ x_{i,s}(k) =1  $, the data rate over link $ (i,\textchange{N(i)}) $ is equal to $ r(k,s,p) $.
	Similarly, constraints \eqref{loops}, \eqref{mediacons13}, and \eqref{mediacons16}
	require that: if cloud node $i$ hosts function $ f_{s+1}^k $, i.e., $ x_{i,s+1}(k) =1  $,
	node $i$'s incoming link  must be assigned for the $ p $-th path of flow $(k,s)$ and
	its outgoing link cannot be assigned for this path; and the data rate over the
	incoming link is equal to $ r(k,s,p) $.
	It is worth remarking that, in contrast to constraints \eqref{mediacons4} and \eqref{mediacons6}, we cannot present constraints \eqref{mediacons15} and \eqref{mediacons16} as 
	\begin{align}
		z_{{N(i)},i}(k, s, p)- z_{i,{N(i)}}(k, s,  p)= -x_{i,s}(k), \tag{\rm{51'}}\label{tmp1}     \\
		z_{{N(i)},i}(k, s, p) - z_{i,{N(i)}}(k, s,  p) = x_{i,s+1}(k). \tag{\rm{52'}} \label{tmp2}
	\end{align}
	Indeed, cloud node $ i $ can potentially process functions $ f_s^k $ or $ f_{s+1}^k $. When cloud node $ i $ processes function $ f_s^k $, i.e., $ x_{i,s}(k) = 1 $, the left-hand side of constraint \eqref{tmp1} must be $ -1 $. However, by constraint \eqref{key} and $ x_{i,s}(k) = 1 $, we have $ x_{i,s+1}(k) = 0 $, which, together with constraint \eqref{tmp2}, further implies the left-hand side of constraint \eqref{tmp1} must be $ 0 $. This is a contradiction.

	We next present the flow conservation constraints of the first and last functions of each flow, which are slightly different from constraints \eqref{mediacons11}-\eqref{mediacons16} due to the fact that $ S(k),D(k) \notin \mathcal{\bar{V}}$. Specifically, for all $ k \in \mathcal{K} $, $ p \in \mathcal{P} $, and $ i \in \mathcal{\barr{I}} $, we have
	\begin{equation*}
		\sum_{j: (j,i) \in \mathcal{\barr{L}}} r_{j,i}(k, 0, p) - \sum_{j: (i,j) \in \mathcal{\barr{L}}} r_{i,j}(k, 0,  p)\qquad \qquad\qquad\qquad
	\end{equation*}
	\begin{numcases}{\!\!\!\!\!\!\!\!\!\!=}
		-r(k, 0, p),     & \text{if}~$i = S(k);$ \label{firstcons11} \\
		0,                                   & \text{if}~$ i \in \bar{\mathcal{I}}\backslash (\bar{\mathcal{V}}\cup \{S(k)\});$  \label{firstcons12}   \\
		r(k, 0, p) x_{i,1}(k),   &  \text{if}~$i \in \bar{\mathcal{V}};$  \label{firstcons13}
	\end{numcases}
	\begin{equation*}
		\sum_{j: (j,i) \in \mathcal{\barr{L}}} z_{j,i}(k, 0, p) - \sum_{j: (i,j) \in \mathcal{\barr{L}}} z_{i,j}(k, 0,  p)\qquad \qquad\qquad\qquad
	\end{equation*}
	\begin{numcases}{\!\!\!\!\!\!\!\!\!\!=}
		-1,     & \qquad\qquad\text{if}~$i = S(k);$\label{firstcons14}   \\
		0,                                   & \qquad\qquad\text{if}~$ i \in \bar{\mathcal{I}}\backslash (\bar{\mathcal{V}}\cup \{S(k)\}) ;$  \label{firstcons15}  \\
		x_{i,1}(k)   &  \qquad\qquad\text{if}~$i \in \bar{\mathcal{V}}$.    \label{firstcons16}
	\end{numcases}
	\noindent For all $ k \in \mathcal{K} $, $ p \in \mathcal{P} $, and $ i \in \mathcal{\barr{I}} $, we have
	\begin{equation*}
		\sum_{j: (j,i) \in \mathcal{\barr{L}}} r_{j,i}(k, {\ell_k}, p) -\sum_{j: (i,j) \in \mathcal{\barr{L}}} r_{i,j}(k, {\ell_k},  p)\qquad \qquad\qquad
	\end{equation*}
	\begin{numcases}{\!\!\!=}
		-r(k, {\ell_k},p)x_{i, {\ell_k}}(k),     & \text{if}~$ i \in \bar{\mathcal{V}};$ \label{lastcons11} \\
		0,                             & \text{if}~$ i \in    \bar{\mathcal{I}}\backslash(\bar{\mathcal{V}}\cup \{D(k)\});$  \label{lastcons12}   \\
		r(k, {\ell_k},p),   &  \text{if}~$i=D(k);$  \label{lastcons13}
	\end{numcases}
	\begin{equation*}
		\sum_{j: (j,i) \in \mathcal{\barr{L}}} z_{j,i}(k, {\ell_k}, p) -\sum_{j: (i,j) \in \mathcal{\barr{L}}} z_{i,j}(k,  {\ell_k},  p)\qquad \qquad\qquad
	\end{equation*}
	\begin{numcases}{\!\!\!=}
		-x_{i, {\ell_k}}(k),    \qquad \qquad   & \text{if}~$ i \in \bar{\mathcal{V}};$\label{lastcons14}   \\
		0,                              \qquad \qquad      & \text{if}~$i \in \bar{\mathcal{I}}\backslash(\bar{\mathcal{V}}\cup \{D(k)\});$  \label{lastcons15}  \\
		1,    \qquad \qquad\ &  \text{if}~$i=D(k)$.    \label{lastcons16}
	\end{numcases}\vspace{-0.1cm}\\
	{\bf\noindent$\bullet$ A New Compact Problem Formulation\vspace{0.1cm}}
	
	Now, we are ready to present the new formulation for the network slicing problem:
	\begin{align}
		  & \min_{\boldsymbol{x},\boldsymbol{y},\boldsymbol{z},\boldsymbol{r},\boldsymbol{\theta}} &   & \sum_{v \in \mathcal{V}}y_v + \sigma \sum_{k \in \mathcal{K}} (\theta_L(k) + \theta_N(k)) \nonumber                                   \\
		  & ~~~~~{\text{s.t.~}}                                                                    &   & \eqref{key}-\eqref{totalcapacitiescons}, ~\eqref{linkdelaycons}-\eqref{delayconstraint},~\eqref{relalambdaandx11}-\eqref{lastcons16}.
		\label{newmip}
		\tag{\rm{NS-II}}
	\end{align}
	Problem \eqref{newmip} is also an MBLP, since the nonlinear terms $ r(k, s,p) x_{i,s}(k) $, $ r(k, s,p) x_{i,s+1}(k) $, $  r(k, 0, p)x_{i,1}(k) $, and $ r(k, {\ell_k},p)x_{i, {\ell_k}}(k) $ in \eqref{mediacons12}, \eqref{mediacons13}, \eqref{firstcons13}, and  \eqref{lastcons11} can be equivalently linearized \cite{Glover1975}.
	Let us take $r(k, s,p) x_{i,s}(k) $ as an example. To linearize $r(k, s,p) x_{i,s}(k)$, we need to introduce an auxiliary variable $ \omega(k,s,p,i) :=  r(k, s,p) x_{i,s}(k) $. From \eqref{relalambdaandx11}, we know $ 0\leq r(k, s,p)  \leq \textchange{\lambda_s(k)} $, which, together with $ x_{i,s}(k) \in \{0,1\} $, implies that
	\begin{equation*}
		0 \leq \omega(k,s,p,i)  \leq r(k, s,p)  \leq \textchange{\lambda_s(k)}.
	\end{equation*}
	We then add the following constraints into the problem: 
	\begin{align*}
		  & \omega(k,s,p,i) \geq 0,                                                                         
		\\
		  & \omega(k,s,p,i) \geq \textchange{\lambda_s(k)}x_{i,s}(k) + r(k,s,p) -\textchange{\lambda_s(k)}, 
		\\            	
		  & \omega(k,s,p,i) \leq  \textchange{\lambda_s(k)}x_{i,s}(k) ,                                     
		\\     
		  & \omega(k,s,p,i) \leq  r(k,s,p).                                                                 
	\end{align*}
	The above four constraints ensure that if $ x_{i,s}(k) =1 $, $ \omega(k,s,p,i) = r(k,s,p) $; otherwise $ \omega(k,s,p,i) = 0  $.

	\subsection{Equivalence of Formulations \eqref{mip} and \eqref{newmip}}
	
	In this subsection, we will show, {somewhat surprising}, that formulations \eqref{mip} and \eqref{newmip} are equivalent, although they are derived in different ways and they take different forms.
	\textchange{The equivalence means that the two formulations share the same optimal solution of the network slicing problem.
		More specifically, for any given feasible point $(x,y,z,r,\theta)$  of formulation \eqref{newmip}, we can construct a feasible solution $(X,Y,Z,R,\Theta) $ of formulation \eqref{mip} such that the two formulations have the same objective value at the corresponding points and vice versa.}
	
	\begin{theorem}
		\label{equivalenceresult}
		Formulations \eqref{mip} and \eqref{newmip} are equivalent.
	\end{theorem}
	\begin{proof}
		In this proof, we shall use $ (X,Y,Z,R,\Theta) $ and $(x,y,z,r,\theta)$ to denote the feasible points of formulations \eqref{mip} and \eqref{newmip}, respectively, in order to differentiate the feasible points of the two formulations.
		We shall show that there exists a one-to-one correspondence between the feasible points of the two formulations.
		More specifically, we shall show that, given any feasible point $(x,y,z,r,\theta)$  of formulation \eqref{newmip}, we can construct a feasible solution $(X,Y,Z,R,\Theta) $ of formulation \eqref{mip} \textchange{such that the two formulations have the same objective value at the corresponding solutions} and vice versa.

		We first show that, given a feasible solution $ (x,y,z,r,\theta) $ of formulation \eqref{newmip}, we shall construct a solution $ (X,Y,Z,R,\Theta) $ of formulation \eqref{mip} as follows: 
		\begin{itemize}
			\item 
			      $ X=x $, $ Y=y $, $ \Theta = \theta $;
			\item 
			      $R(k,s,v_s,v_{s+1}, p)=$
			      \begin{numcases}{\!\!\!\!\!\!\!\!\!\!\!}
				      r(k,s,p), & ~~if $ x_{v_s,s}(k) =  x_{v_{s+1},s+1}(k) = 1  $;\label{trans1}\\
				      0, & ~~otherwise;\label{trans2}
			      \end{numcases}
			\item 
			      $Z_{i,j}(k,s,v_s,v_{s+1}, p)=$
			      \begin{numcases}{\!\!\!\!\!\!\!\!\!}
				      z_{i,j}(k,s,p), & if $ x_{v_s,s}(k) =  x_{v_{s+1},s+1}(k) = 1;  $ \label{trans3}\\
				      0, & otherwise;\label{trans4}
			      \end{numcases}
			\item $ R_{i,j}(k,s,v_s,v_{s+1}, p)= $
			      \begin{numcases}{\!\!\!\!\!\!\!\!\!}
				      r_{i,j}(k,s,p), & if $ x_{v_s,s}(k) =  x_{v_{s+1},s+1}(k) = 1  $;\label{trans5}\\
				      0, & otherwise. \label{trans6}
			      \end{numcases}
		\end{itemize}
		We need to take additional care of the above constructions \eqref{trans1}-\eqref{trans6} when $ s=0 $ and $ s=\ell_k $. In particular, when $ s=0 $, we let $ v_s = S(k) $ and $ x_{v_s,s} = 1 $; when $ s=\ell_k $, we let $ v_{s+1} = D(k) $ and $ x_{v_{s+1},s+1} = 1 $. 
		
		Based on the above construction, we have the following key relationships, which relate the feasible points of the two problems: 
		\begin{align}
			r(k,s,p) =\sum_{v_s, v_{s+1} \in \bar{\mathcal{V}}} R(k,s,v_s,v_{s+1},p),\label{relat1} ~~          \\
			z_{i,j}(k,s,p) =\sum_{v_s, v_{s+1} \in \bar{\mathcal{V}}} Z_{i,j}(k,s,v_s,v_{s+1},p),\label{relat2} \\
			r_{i,j}(k,s,p)=\sum_{v_s, v_{s+1} \in \bar{\mathcal{V}}} R_{i,j}(k,s,v_s,v_{s+1},p). \label{relat3}
		\end{align}
		\textchange{Clearly, formulations \eqref{mip} and \eqref{newmip} have the same objective value at points $  (X,Y,Z,R,\Theta) $ and $ (x,y,z,r,\theta) $, respectively.}
		We \textchange{next} show that the constraints in formulation \eqref{mip} hold true at point $  (X,Y,Z,R,\Theta) $.
		Obviously, constraints \eqref{key}-\eqref{totalcapacitiescons} and  \eqref{linkdelaycons}-\eqref{delayconstraint} hold at this point.
		Combining \eqref{relat2} and \eqref{consdelay2funs1}, we know that constraint \eqref{consdelay2funs} also holds. The link capacity constraint \eqref{linkcapcons} follows from constraint \eqref{linkcapcons1} and Eq. \eqref{relat3}. 
		
		We now prove that constraints \eqref{relalambdaandx1}-\eqref{relarandz1} are satisfied at point $  (X,Y,Z,R,\Theta) $.
		Combining Eqs. \eqref{trans1}-\eqref{trans2} and \eqref{relalambdaandx11} shows that constraints \eqref{relalambdaandx1}-\eqref{relalambdaandx3} are satisfied.
		From Eqs. \eqref{trans3}-\eqref{trans4}, it follows that constraints \eqref{relazandx1}-\eqref{relazandx3} hold.
		By Eqs.  \eqref{trans3}-\eqref{trans4} and \eqref{loops}, we know that constraint \eqref{loops-1}-\eqref{loops-2} holds.
		Using Eqs. \eqref{trans3}-\eqref{trans6} and  \eqref{relarandz11}, we can obtain that constraint \eqref{relarandz1} is satisfied.

		Next we show that the flow conservation constraints \eqref{mediacons1}-\eqref{lastcons6} hold at $  (X,Y,Z,R,\Theta) $.
		We only need to consider the case where $ x_{v_s, s}(k) = x_{v_{s+1}, s+1}(k) =1 $ (or $  x_{v_1, 1}(k) = 1 $ and $ x_{v_{\ell_k},\ell_k}(k) =1$) since all the other cases are trivial to prove.
		Using the relations \eqref{trans1}-\eqref{trans6} and the flow conservation constraints \eqref{firstcons11}-\eqref{lastcons16}, it is simple to show that constraints \eqref{firstcons1}-\eqref{lastcons6} are satisfied.
		The proof of showing that point $  (X,Y,Z,R,\Theta) $ satisfies constraints \eqref{mediacons2} and \eqref{mediacons5} is similar. 
		It remains to show that constraints  \eqref{mediacons1}, \eqref{mediacons3}, \eqref{mediacons4}, and \eqref{mediacons6} hold.
		From \eqref{mediacons15} and \eqref{trans3}, we have 
		$$
			Z_{v_s,\textchange{N(v_s)}}(k, s, v_s, v_{s+1}, p) = z_{v_s,\textchange{N(v_s)}}(k, s, p)  =1.
		$$
		This, together with \eqref{loops} and \eqref{trans3}, implies 
		$$
			Z_{\textchange{N(v_s)},v_s}(k, s, v_s, v_{s+1}, p) =z_{\textchange{N(v_s)},v_s}(k, s, p)  =0.
		$$
		Consequently, constraint \eqref{mediacons4} holds.
		Similarly, we can show that constraint \eqref{mediacons6} holds true.
		Finally, it follows from \eqref{relarandz11}, \eqref{mediacons12}, \eqref{mediacons13}, \eqref{trans1}, and \eqref{trans5} that constraints \eqref{mediacons1} and \eqref{mediacons3} also hold true.
		
		We now prove the other direction. Given a feasible solution $  (X,Y,Z,R,\Theta) $ of formulation \eqref{mip}, we  construct a solution $ (x,y,z,r,\theta) $ by setting $ x= X $, $ y=Y $, $ \theta=\Theta $, and \eqref{relat1}-\eqref{relat3}. Using the previous arguments, we can show that $  (x,y,z,r,\theta)  $ satisfies all constraints in formulation \eqref{newmip}. 
	\end{proof}
	
	Theorem \ref{equivalenceresult} shows that there is a one-to-one correspondence
	between the feasible points of formulations \eqref{mip} and \eqref{newmip},
	and all information of solving the higher-dimensional formulation \eqref{mip} can be
	obtained by solving a more compact lower-dimensional formulation \eqref{newmip}. 
	To be specific, both of the numbers of  variables and constraints in formulation \eqref{newmip} are 
$ \mathcal{O}(|\mathcal{\bar{L}}||\mathcal{P}|\sum_{k \in \mathcal{K}}\ell_k)$.
	This is significantly less than those in formulation \eqref{mip}, which are
${\cal O}({|\mathcal{\bar{V}}|}^2|\mathcal{\bar{L}}||\mathcal{P}|\sum_{k \in \mathcal{K}}\ell_k )$.
	Therefore, formulation \eqref{newmip} can be much easier to solve than formulation \eqref{mip}, as demonstrated in the next section.
	{In addition, as shown in Proposition \ref{NPhard}, problem \eqref{mip} (and thus problem \eqref{newmip}) is strongly NP-hard, meaning that low-complexity algorithms for approximately solving the network slicing problem is needed in practice, especially when the problem's dimension is large.
	We remark that the compact formulation \eqref{newmip} is an important step towards developing an efficient low-complexity algorithm for solving the network slicing problem due to the following two reasons.
	First, globally solving the problem provides an important benchmark for evaluating the solution quality of the designed low-complexity (heuristic online) algorithm.
	Second, compact formulation \eqref{newmip} will often lead to a more compact linear programming (LP) relaxation (as compared with the LP relaxation of \eqref{mip}), which is actually the basis of many low-complexity LP relaxation based algorithms.
	In the future, we will develop low-complexity algorithms based on the LP relaxation of formulation \eqref{newmip} to extend the work in this paper.
	}
	
	\section{Numerical Results}
	\label{subsec:experiments}
	In this section, we present numerical results to compare the performance of our proposed problem formulations \eqref{mip} and \eqref{newmip}, and the existing problem formulations in \cite{Zhang2017} and \cite{Woldeyohannes2018}.
	In particular, we first perform numerical experiments to compare the computational efficiency of formulations \eqref{mip} and \eqref{newmip}.
	Then, we present some simulation results to demonstrate the effectiveness of our proposed formulation \eqref{newmip} over the state-of-the-art formulations in \cite{Zhang2017} and \cite{Woldeyohannes2018}.
	{Finally, we evaluate the performance of our proposed formulation \eqref{newmip} under different network parameters.}
	
	In both formulations \eqref{mip} and \eqref{newmip}, unless otherwise stated, we choose $\sigma = 0.001$ (which satisfies condition in Proposition \ref{oneproblemcondition}) and $P=2$.
	All MBLP problems are solved using Gurobi 9.0.1\footnote{{The Matlab codes of the two formulations can be downloaded at \url{https://github.com/chenweikun/networkslicing}.}} \cite{Gurobi}. We set a time limit of 600 seconds for  Gurobi, i.e., we terminate the solution process if the CPU time is over 600 seconds.
	\textchange{The relative gap tolerance of Gurobi is set to be 0.1\%, i.e., a feasible solution which has an optimality gap of 0.1\% is considered to be optimal.}
	All experiments were performed on a server with 2 Intel Xeon E5-2630 processors and 98 GB of RAM, using the Ubuntu GNU/Linux Server 14.04 x86\_64 operating system.

	\subsection{Comparison of Formulations \eqref{mip} and \eqref{newmip}}
	\label{compformsec}
	In this subsection, we compare \textchange{the computational efficiency} of solving the two equivalent
	formulations \eqref{mip} and \eqref{newmip} on some small randomly generated
	networks. The randomly generated procedure is described as follows. We randomly
	generate a network consisting of $ 6 $ nodes on a $ 100\times100 $ region in the
	Euclidean plane including $ 3 $ cloud nodes.
	We generate link $ (i,j) $ for each pair of nodes $ i $ and $ j $ with the probability of $0.6$.
	The communication delay on link $ (i,j) $ is calculated by the distance of link $ (i,j) $ over $ \bar{d} $, where $ \bar{d} $ is the average length of all shortest paths between every pair of nodes.
	The cloud node and link capacities are randomly chosen in $ [6,12] $ and $ [0.5,3.5] $, respectively.
	There are in total $ 5 $ different service functions, i.e., $ \{f^1, \ldots, f^5\} $.
	Among the $3$ cloud nodes, $2$ cloud nodes are randomly chosen to process $ 2 $ service functions of $ \{f^1, \ldots, f^5\} $ and the remaining one is able to process all the service functions.
	The processing delay of each function in each cloud node is randomly chosen in $ [0.8,1.2] $.
	For each service $ k $, nodes $ S(k) $ and $ D(k) $ are randomly chosen from the available network nodes excluding the cloud nodes; SFC $ \mathcal{F}(k) $ is an ordered sequence of functions randomly chosen from $ \{f^1, \ldots, f^5\} $ with $ |\mathcal{F}(k)|=3 $; the service function rates $ \lambda_s(k) $ are all set to $ 1 $; and the E2E delay threshold $ \Theta_k $ is set to $ 3+(6*\text{dist}_k+\alpha) $ where $ \text{dist}_k $ is the shortest path (in terms of the delay) between nodes $ S(k) $ and $ D(k) $ and $ \alpha $ is randomly chosen in $[0,2]$.
	In our simulations, we randomly generate 100 problem instances for each fixed number of
	services and the results presented below are based on statistics from all these 100 problem instances.
	
	\begin{figure}[h]
		\centering
		\includegraphics[height=\figuresize]{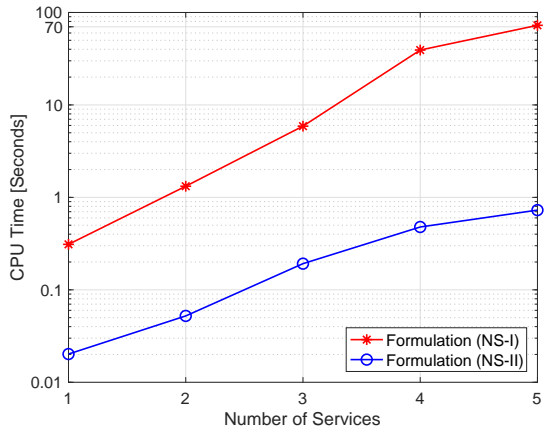}
		\caption{{The CPU time taken by solving formulations \eqref{mip} and \eqref{newmip}.}}
		\label{diffformulaitons}
	\end{figure}
	{Fig. \ref{diffformulaitons} plots the average CPU time taken by solving formulations \eqref{mip} and \eqref{newmip} versus the number of services.
	From Fig. \ref{diffformulaitons}, it can be clearly seen that it is much more efficient to solve formulation \eqref{newmip} than formulation \eqref{mip}.
	In particular, when the number of services is equal to 5, the CPU time of solving formulation \eqref{mip} is more than 70 seconds while that of solving formulation \eqref{newmip} is less than 1 second.
	From this simulation result, we can conclude that formulation \eqref{newmip} significantly outperforms formulation \eqref{mip} in terms of the solution efficiency.
	Due to this, we shall only use and discuss formulation \eqref{newmip} in the following.}

	\subsection{Comparison of Proposed Formulation \eqref{newmip} and Those in \cite{Zhang2017} and \cite{Woldeyohannes2018}}
	\label{compexistingforms}
	
	In this subsection, we present simulation results to illustrate the effectiveness of our proposed formulation compared with those in \cite{Zhang2017} and \cite{Woldeyohannes2018}.
	
	We consider the {\emph{fish network topology}} studied in \cite{Zhang2017}, consisting of 112 nodes and 440 links.
	The network includes 86 nodes that can be potentially chosen as the source node of the flows and only a single node that can be chosen as the destination node of the flows; see \cite{Zhang2017} for more details.  
	There are six cloud nodes that can potentially process service functions: five of them are randomly chosen to process two service functions of $ \{f^1, \ldots, f^4\} $ and the remaining one is chosen to process all the service functions.
	The cloud nodes' capacities are randomly chosen in $ [50,100] $.
	The links' capacities are randomly chosen in $ [7,77]  $.
	The NFV	 and communication delays on the cloud nodes and links are randomly chosen in $\{3,4,5,6\}$ and $\{1,2\}$, respectively.
	For each service $k$, node $S(k)$ is randomly chosen from the 86 available nodes and node $D(k)$ is set to be the common destination node; SFC $ \mathcal{F}(k) $ is an ordered sequence of functions randomly chosen from $ \{f^1, \ldots, f^4\} $ with $ |\mathcal{F}(k)|=3 $; the service function rate $ \lambda_s(k) $ are all set to be the same integer value, randomly chosen from $ [1,11] $; the E2E delay threshold $ \Theta_k $ is set to $ 20+(3*\text{dist}_k+\alpha) $ where $ \text{dist}_k $ is the shortest path (in terms of the delay) between nodes $ S(k) $ and $ D(k) $ and $ \alpha $ is randomly chosen in $[0,5]$. The above parameters are carefully chosen such that the constraints in the network slicing problem are neither too tight nor too loose. For each fixed number of services, we randomly generate 100 problem instances and the results presented below are based on statistics from all these 100 problem instances.
	
	\begin{figure}[h]
		\centering
		\includegraphics[height=\figuresize]{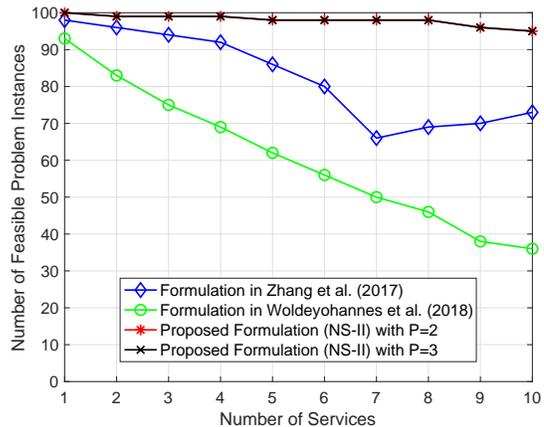}
		\caption{{Number of feasible problem instances by solving the formulations in \cite{Zhang2017}, \cite{Woldeyohannes2018} and our proposed formulation \eqref{newmip}.}}
		\label{feascases}
	\end{figure}
	
	Fig. \ref{feascases} plots the number of feasible problem instances by solving the formulations in \cite{Zhang2017}, \cite{Woldeyohannes2018} and our proposed formulation \eqref{newmip} with $P=2$ and $P=3$, where $ P $ is the maximum number of paths allowed to route the traffic flow between any pair of cloud nodes that process two adjacent functions of a service.
	Since the formulation in \cite{Zhang2017} does not explicitly take the E2E latency constraints into consideration, the blue-diamond curve in Fig. \ref{feascases} is obtained as follows.
	We solve the formulation in \cite{Zhang2017}
	and then substitute the obtained solution into the E2E latency constraints in \eqref{delayconstraint}: if the solution satisfies all E2E latency constraints, we count the corresponding problem instance feasible; otherwise it is infeasible.
	{As for the curves of our proposed formulation \eqref{newmip} with $P=2$ and $P=3$ and the formulation in \cite{Woldeyohannes2018}, since the E2E latency constraints \eqref{delayconstraint} are explicitly enforced, we solve the formulation by Gurobi and count the corresponding problem instance feasible if Gurobi can return a feasible solution; otherwise it is infeasible.
	Notice that due to the random generation procedure of the problem instances, we cannot guarantee that every problem instance is feasible.}
	
	We first compare the performance of formulation \eqref{newmip} with $P=2$ and $P=3$.
	Intuitively, the number of feasible problem instances by solving \eqref{newmip} with $P=3$ should be larger than that of using \eqref{newmip} with $P=2$ as there is more traffic routing flexibility in the formulation with $P=3$.
	However, as can be seen from Fig. \ref{feascases}, solving formulation \eqref{newmip} with $P=2$ and $P=3$ gives the same number of feasible problem instances.
	This sheds a useful insight that there is already a sufficiently large flexibility of traffic routing in formulation \eqref{newmip} with $P=2$.
	Therefore, in practice we can simply set $P=2$ in formulation \eqref{newmip}. In
	some cases where formulation \eqref{newmip} with $P=2$ does not have a
	satisfactory performance (e.g., the total energy consumption or the E2E latency of the
	service), we can increase $P$ to potentially increase the flexibility of traffic routing
	with a sacrifice of the solution efficiency.  

	Next, we compare our problem formulation \eqref{newmip} with the formulations in \cite{Zhang2017} and \cite{Woldeyohannes2018}.
	First, from Fig. \ref{feascases}, the flexibility of traffic routing in our proposed
	formulation \eqref{newmip} allows for solving a much larger number of problem
	instances than that can be solved by using the formulation in
	\cite{Woldeyohannes2018} (which can be seen as a special case of our formulation
	\eqref{mip}, or equivalently, formulation \eqref{newmip}, with $P=1$, as discussed at the end of
	Section \ref{naturalformulation}), especially in the case where the number of services
	is large.
	For instance,  when the number of services is $10$, $95 $ problem instances
	(from a total of $100$) are feasible by solving our formulation \eqref{newmip},
	while only 36 problem instances are feasible by solving the formulation in \cite{Woldeyohannes2018}.
	Second, it can be seen from Fig. \ref{feascases} that the number of feasible problem instances of solving our proposed formulation \eqref{newmip} is also  larger than that of solving the formulation in \cite{Zhang2017}.
	This clearly shows the advantage of our proposed formulation (i.e., it has a guaranteed E2E Latency) over that in \cite{Zhang2017}.
	In summary, the results in Fig. \ref{feascases} illustrate that, compared with those in \cite{Zhang2017} and \cite{Woldeyohannes2018}, our proposed formulation gives a ``better'' solution.
	More specifically, compared with that in \cite{Zhang2017}, our formulation has a \emph{guaranteed} E2E delay; compared with that in \cite{Woldeyohannes2018}, our formulation allows for \emph{flexible} traffic routing.
	
	\subsection{Evaluation of Proposed Formulation \eqref{newmip}}

	To gain more insight into the performance of formulation \eqref{newmip}, we carry out
	more numerical experiments on problem instances with different link capacities.
	More specifically, we generate another set of problem instances using the same randomly generated procedure, as in Section \ref{compexistingforms}, except that the links' capacity is randomly chosen in $[5,55]$.
	We refer this set as ``Low Link Capacity'' and compare it with the set in
	Section \ref{compexistingforms}, which is referred as ``High Link Capacity''.

	\begin{figure}[h]
		\centering
		\includegraphics[height=\figuresize]{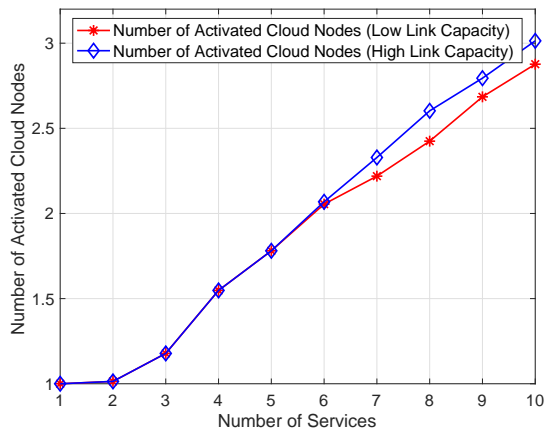}
		\caption{{The number of activated cloud nodes in formulation \eqref{newmip}}.}
		\label{activatednodesres1}
	\end{figure}
	{\bf\noindent$\bullet$ Number of Activated Cloud Nodes\vspace{0.15cm}}\\
	\indent Fig. \ref{activatednodesres1} shows the average number of activated cloud nodes in the physical network.
	We can observe that, as expected, for both sets, more cloud nodes need to be activated as the number of services increases.
	In addition, when the number of services is small (e.g., $|\mathcal{K}| \leq 6$), the numbers of activated cloud nodes in the two sets are almost the same.
	However, when the number of services is large (e.g., $|\mathcal{K}| \geq 7$), more cloud nodes need to be activated in the problem set with low link capacity, compared to that with high link capacity.
	This can be explained as follows. As the number of the services increases, the traffic in the network  becomes heavier.
	The latter further results in the situation that some activated cloud node cannot
	process the functions in some services due to the reason that some links'
	capacities are not enough to route the data flow.
	Therefore, more cloud nodes generally need to be activated in the problem set with low link capacity{, which leads to larger power consumption. This highlights an interesting tradeoff between the communication capacity and the power consumption of the underlying cloud network.} \vspace{0.1cm}\\
	\begin{figure}[t]
		\centering
		\includegraphics[height=\figuresize]{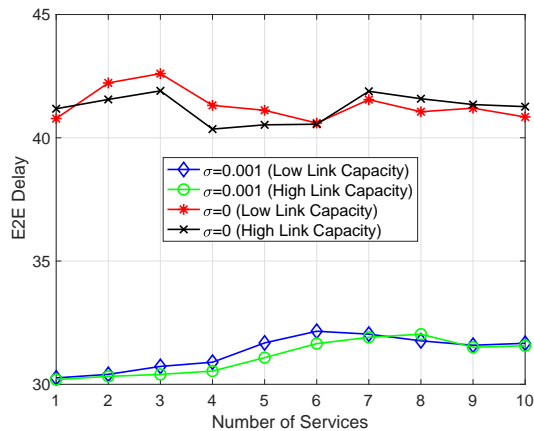}
		\caption{{The average E2E delay in problem \eqref{newmip}.}}
		\label{delaysres1}
	\end{figure}
	{\bf\noindent$\bullet$ E2E Delay\vspace{0.1cm}}\\
	\indent We now consider the E2E delays of the services. Fig. \ref{delaysres1} plots the average E2E delay by solving formulation \eqref{newmip} with $\sigma = 0.001$ and $\sigma= 0$.
	From the figure, we observe that for both ``High Link Capacity'' and ``Low Link Capacity'' cases, the average E2E delay by solving formulation \eqref{newmip} with $\sigma = 0.001$ is much smaller than that with $\sigma =0$.
	This clearly shows the advantage of adding the total delay term \eqref{delobj} into the objective in formulation \eqref{newmip}.
	In addition, for formulation \eqref{newmip} with $\sigma=0.001$, comparing the E2E delays of the two cases when the numbers of activated cloud nodes are the same (i.e., the number of services is less than or equal to $ 6$; see Fig. \ref{activatednodesres1}), the average E2E delay in ``High Link Capacity'' is slightly smaller than that in ``Low Link Capacity''. This is because with larger links' capacities, a service can potentially choose a path to transmit its data flow with a smaller E2E delay. \vspace{0.15cm}\\
	\begin{figure}[h]
		\centering
		\includegraphics[height=\figuresize]{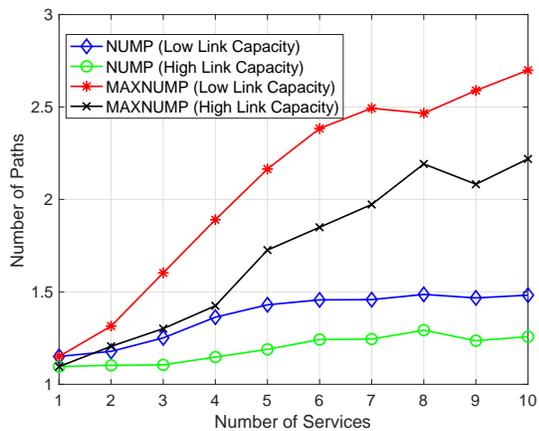}
		\caption{{The number of used paths in problem \eqref{newmip}.}}
		\label{pathsres1}
	\end{figure}
	\begin{figure}[h]
		\centering
		\includegraphics[height=\figuresize]{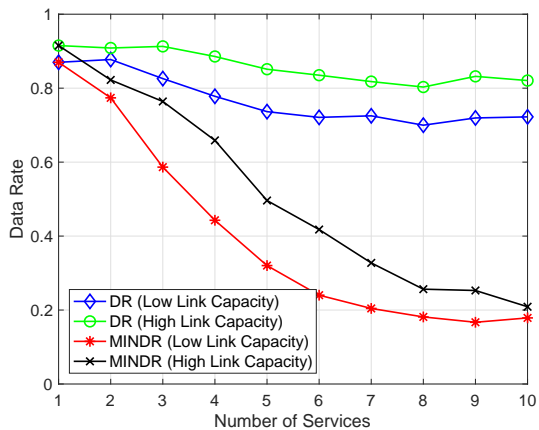}
		\caption{{The data rate on the paths in problem (NS-II).}}
		\label{datarate1}
	\end{figure}
	
	{\bf\noindent$\bullet$ Flexibility of Traffic Routing and the Number of Used Paths \vspace{0.15cm}}\\
	\indent
	{Finally, we consider the flexibility of traffic routing by comparing the number of paths of the services used to route the traffic from their source nodes to destination nodes\footnote{{Notice that setting $P=2$ in our formulation means that the number of paths is at most $2$ between the nodes who provide two adjacent functions of each service. The number of paths of each service used to route the traffic from its source node to its destination node lies in the interval $[1,\ell_k+1]$.}}
		and the data rates on the corresponding paths.
		Specifically, after solving each problem instance, we calculate the minimum number
		of paths \cite{Hartman2012} and the corresponding minimum data rate on these paths, denoted by {NUMP} and {DR}, needed to realize the routing strategy of the traffic flow of each service.
		For each problem instance, let MAXNUMP and MINDR denote the maximum NUMP and the minimum DR among all the services, respectively.

		Figs. \ref{pathsres1} and \ref{datarate1} plot the results of average NUMP and MAXNUMP and average DR and MINDR, respectively.
		In general, as the number of services increases, both NUMP and MAXNUMP become larger,
		which shows that more paths are used to carry out the traffic flow of the services and, as a result, their data rates (DR or MINDR) on the paths are likely to become smaller, as observed in Fig. \ref{datarate1}.
		In addition, Fig. \ref{pathsres1} clearly shows that, compared to the problem instances with low link capacity, the problem instances with high link capacity need fewer paths to route their traffic flow, and consequently, their data rates on the paths are likely to become larger, as observed in Fig. \ref{datarate1}.
		This shows that the heavier the traffic is or the smaller the link capacity is, the more flexibility of traffic routing generally is exploited and used in our proposed formulation \eqref{newmip}.}

	\section{\textchange{Conclusions and future work}}
	\label{conclusions}
	In this paper, we have investigated the network slicing problem that plays a crucial role in the 5G and B5G networks.
	We have proposed two new MBLP formulations for the network slicing problem, which can be optimally solved by {standard solvers} like Gurobi.
	Our proposed formulations minimize a weighted sum of the total power consumption
	of the whole cloud network (equivalent to the number of activated cloud nodes) and the
	total delay of all services {subject to the SFC constraints, the E2E latency
			constraints of all services, and all cloud nodes' and links' capacity constraints.}
	While we show that our proposed two formulations are mathematically equivalent, the second formulation, when compared to the first one, has a significantly smaller number of variables and constraints, which makes it much more efficiently solvable.
	Numerical results demonstrate the advantage of our proposed formulations over the existing ones in \cite{Zhang2017} and \cite{Woldeyohannes2018}.
	In particular, in addition to a guaranteed E2E latency of all services, our
	proposed formulations (even with $P=2$) offer a large degree of freedom of flexibly
selecting paths to route the traffic flow of all services from their source nodes to
destination nodes and thus can effectively alleviate the effects of the limited link
capacity on the performance of the whole network. 
{In addition, our analysis shows an interesting tradeoff between the communication capacities of the links and the power consumption of the underlying cloud network.}

As the future work, we shall develop low-complexity {LP relaxation based} algorithms for efficiently solving the network slicing  problem and possibly analyze the approximation performance of the algorithms.
{Tight LP relaxations and effective rounding techniques (which round the fractional solutions into desired binary solutions) are crucial to the performance of low-complexity LP relaxation based  algorithms. To develop tight LP relaxations and effective rounding techniques, we need to judiciously exploit the special structure of the problem.}
Some recent progress along this direction has been reported in \cite{Chen2021}. 
{In addition, it is interesting to consider other practical factors such as privacy \cite{Sun2019,Reyhanian2020} and the more practical nonlinear NFV delay model \cite{Baumgartner2015} in the network slicing problem formulation.}

\end{document}